\documentclass[preprint,12pt]{elsarticle}

\usepackage{proof}
\usepackage{amssymb,amsmath,amsthm}
\usepackage{times}
\usepackage{color}
\usepackage{latexsym}

\usepackage[all]{xy}

\journal{Theoretical Computer Science}

\newcommand{\FV}[1]{{\sf FV}(#1)}
\newcommand{\FTV}[1]{{\sf FTV}(#1)}

\newcommand{\lb}{\lambda}
\newcommand{\Lb}{\Lambda}
\newcommand{\llb}{\lb\hspace{-5.5pt}\lb}
\newcommand{\pair}[2]{\langle#1,#2\rangle}
\newcommand{\ipair}[2]{\langle #2\rangle_{#1=1,2}} 
\newcommand{\ppair}[2]{\langle\hspace{-3pt}\langle#1,#2\rangle\hspace{-3pt}\rangle}
\newcommand{\proj}[2]{#2#1}
\newcommand{\pproj}[2]{#2@#1} 
\newcommand{\injn}[4]{\mathsf{in}_{#1}(#2,#3,#4)} 
\newcommand{\ainjn}[4]{\underline{\mathtt{in}}_{#1}(#2,#3,#4)}
\newcommand{\aiinjn}[4]{\underline{\mathtt{IN}}_{#1}(#2,#3,#4)} 
\newcommand{\ainjsymb}{\underline{\mathsf{in}}} 
\newcommand{\aiinjsymb}{\underline{\mathsf{IN}}} 
\newcommand{\case}[6]{\mathsf{case}(#1,#2.#3,#4.#5,#6)}
\newcommand{\casesymb}{\mathsf{case}} 
\newcommand{\cse}[5]{\mathsf{case}(#1,#2.#3,#4.#5)} 
\newcommand{\acasesymb}{\underline{\mathtt{case}}} 
\newcommand{\acase}[6]{\underline{\mathtt{case}}(#1,#2.#3,#4.#5,#6)} 
\newcommand{\accase}[6]{\underline{\mathtt{CASE}}(#1,#2.#3,#4.#5,#6)} 
\newcommand{\accasesymb}{\underline{\mathtt{CASE}}} 

\newcommand{\abort}[2]{\mathsf{abort}(#1,#2)}
\newcommand{\aabortsymb}{\underline{\mathtt{abort}}} 
\newcommand{\aabort}[2]{\underline{\mathtt{abort}}(#1,#2)} 
\newcommand{\aabbortsymb}{\underline{\mathtt{ABORT}}} 
\newcommand{\aabbort}[2]{\underline{\mathtt{ABORT}}(#1,#2)} 

\newcommand{\connective}{\bigcirc} 

\newcommand{\dvee}{\underline{\vee}} 
\newcommand{\dperp}{\underline{\perp}} 

\newcommand{\betai}{\beta_{\supset}}
\newcommand{\betac}{\beta_{\wedge}}
\newcommand{\betad}{\beta_{\vee}} 
\newcommand{\betaall}{\beta_{\forall}} 
\newcommand{\betas}{\beta_{\connective}} 

\newcommand{\pii}{\pi_{\supset}}
\newcommand{\pic}{\pi_{\wedge}}
\newcommand{\pid}{\pi_{\vee}}
\newcommand{\pia}{\pi_{\perp}}
\newcommand{\pis}{\pi_{\connective}}

\newcommand{\abi}{\varpi_{\supset}}
\newcommand{\abc}{\varpi_{\wedge}}
\newcommand{\abd}{\varpi_{\vee}}
\newcommand{\aba}{\varpi_{\perp}}
\newcommand{\abs}{\varpi_{\connective}}

\newcommand{\etas}{\eta_{\connective}}
\newcommand{\etai}{\eta_{\supset}}
\newcommand{\etac}{\eta_{\wedge}}
\newcommand{\etad}{\eta_{\vee}}
\newcommand{\etaall}{\eta_{\forall}}

\newcommand{\ipc}{\mathbf{IPC}}
\newcommand{\fat}{{\mathbf{F}}_{\mathbf{at}}}
\newcommand{\f}{\mathbf{F}}

\newcommand{\am}[1]{#1^{\circ}} 
\newcommand{\om}[1]{#1^{\lozenge}} 
\newcommand{\cm}[1]{#1^{\star}} 
\newcommand{\rpm}[1]{#1^{\bullet}} 

\newcommand{\too}{\twoheadrightarrow}

\newtheorem{defn}{Definition}
\newtheorem{lem}{Lemma}
\newtheorem{cor}{Corollary}
\newtheorem{prop}{Proposition}
\newtheorem{thm}{Theorem}
\newtheorem{con}{Convention}

\begin{document}

\begin{frontmatter}



\title{How to avoid the commuting conversions of $\ipc$}


\author[label1]{Jos\'e Esp\'\i{}rito Santo\fnref{label2}\corref{cor1}}
\ead{jes@math.uminho.pt}
\affiliation[label1]{
       organization={Centro de Matemática},
           addressline={Universidade do Minho}, 
            country={Portugal}
}
\fntext[label2]{The author was partially financed by Portuguese Funds through FCT (Fundação para a Ciência e a Tecnologia) within the Projects UIDB/00013/2020 and UIDP/00013/2020}
        

 \cortext[cor1]{Corresponding author}
 
\author[label3,label4]{Gilda Ferreira\fnref{label5}}
\ead{gmferreira@fc.ul.pt}
\fntext[label5]{The author acknowledges the support of Fundação para a Ciência e a Tecnologia under the projects [UIDB/04561/2020, UIDB/00408/2020 and UIDP/00408/2020] and is also grateful to Centro de Matem\'{a}tica, Aplica\c{c}\~{o}es Fundamentais e Investiga\c{c}\~{a}o Operacional and to LASIGE - Computer Science and Engineering Research Centre (Universidade de Lisboa). }
\affiliation[label3]{organization={DCeT},
	addressline={Universidade Aberta}, 
    city={Lisboa},
	postcode={1269-001}, 
	state={Lisboa},
	country={Portugal}}

\affiliation[label4]{organization={CMAFcIO},
	addressline={Faculdade de Ciências, Universidade de Lisboa}, 
	postcode={1749-016}, 
	state={Lisboa},
	country={Portugal}}

\begin{abstract}
Since the observation in 2006 that it is possible to embed IPC into the atomic polymorphic $\lb$-calculus (a predicative fragment of system $\f$ with universal instantiations restricted to atomic formulas) different such embeddings appeared in the literature. All of them comprise the Russell-Prawitz translation of formulas, but have different strategies for the translation of proofs. Although these embeddings preserve proof identity, all fail in delivering preservation of reduction steps. In fact, 
they translate the commuting conversions of IPC to $\beta$-equality, or to other kinds of reduction or equality generated by new principles added to system $\f$. The cause for this is the generation of redexes by the translation itself.
In this paper, 
we present an embedding of $\ipc$ into atomic system $\f$, still based on the same translation of formulas, but 
which maps commuting conversions to syntactic identity, while simulating the other kinds of reduction steps present in $\ipc$ by 
$\beta\eta$-reduction. 
In this sense the translation achieves a truly commuting-conversion-free image of $\ipc$ in atomic system $\f$.

\end{abstract}




\end{frontmatter}



\section{Introduction}\label{sec:intro}

The renewed interest in the interpretation of Intuitionistic Propositional Calculus ($\ipc$) into second-order logic (or system $\f$) has several motivations: one is to make use of small, predicative fragments of the interpreting logic, like the subsystem $\fat$ with universal instantiation restricted to atomic formulas \cite{Ferreira2006,FerreiraFerreira2013}; another is to obtain an image of $\ipc$ with better ``structural'' properties, like the elimination of commuting conversions \cite{FerreiraFerreira2009}. In fact, the latter goal has been pursued through several different embeddings of $\ipc$ into $\fat$, just making use of atomic polymorphism: in addition to the one in \cite{FerreiraFerreira2009,FerreiraFerreira2013}, based on \emph{instantiation overflow} \cite{DinisFerreira2016}, we are aware of an alternative in \cite{PistoneTranchiniPetrolo2021}, based on a more efficient procedure for instantiation overflow, besides our own proposal in \cite{EspiritoSantoFerreira2020}, based on the admissibility of the elimination rules for disjunction and absurdity. 

System $\fat$ is, in a way, a radically satisfying answer to the quest for a predicative interpreting fragment of $\f$ (but see the discussion in \cite{EspiritoSantoFerreira2021} about the ``timing'' of atomization, or the limitations pointed out in \cite{PistoneTranchiniPetrolo2021} to interpretations based on atomic polymorphism); on the other hand, the second goal of finding an image of $\ipc$ free from commuting conversions has still not been attained by the three referred embeddings into $\fat$. While all of them agree on the (Russell-Prawitz) translation of formulas, they differ on how to translate derivations. While all of them give a satisfying simulation of the $\beta\eta$-reductions of $\ipc$, they are still not totally satisfying regarding the simulation of commuting conversions: in some cases, the translation of redex and \emph{contractum} are related only by $\beta$-equality (as in \cite{EspiritoSantoFerreira2020}), or even by an equality that requires additional axioms (as in \cite{PistoneTranchiniPetrolo2021}).

Lack of preservation of reduction is related to the problem of generation of \emph{administrative} redexes, i.e., redexes that do not result from the translation of redexes already present in the source derivation, but rather are generated by the translation itself. The translations in \cite{PistoneTranchiniPetrolo2021} and \cite{EspiritoSantoFerreira2020} made substantial progress considering the previous problem, in comparison with the pioneer translation in \cite{Ferreira2006,FerreiraFerreira2013}, with the translation by the present authors \cite{EspiritoSantoFerreira2020} being the more economic one in terms of the creation of administrative redexes, which means that it generates smaller proofs, which require lesser reduction steps to normalize.

Even so, our translation in \cite{EspiritoSantoFerreira2020} 
does not completely avoid administrative redexes, and for this reason it suffers from a defect felt even more acutely with the other two translations: given a commuting conversion $M\to N$ in $\ipc$, sometimes we have to reduce some administrative redexes in the images, $\am{M}\to_{\beta}^* M'$ and $\am{N}\to_{\beta}^* N'$, before we can see the reduction $M'\to_{\beta}^+ N'$ that corresponds to the source conversion. The problem is that, in the overall ``simulation'', bridging $\am{M}$ and $\am{N}$, the reduction steps in $\am{N}\to^* N'$ go in the ``wrong direction'' and only a $\beta$-equality results: $\am{M}=_{\beta}\am{N}$.


In this paper we will optimize our translation $\am{(\cdot)}$ from \cite{EspiritoSantoFerreira2020}, at the level of generation of administrative redexes, to obtain a translation $\om{(\cdot)}$ of $\ipc$ into $\fat$ that, while preserving $\beta\eta$-reduction, has the following property: if $M\to N$ is a commuting conversion in $\ipc$, then $\om{M}$ and $\om{N}$ are the same proof in $\fat$. If translation $\am{(\cdot)}$ achieved a quantitative improvement over the pioneer translation in \cite{Ferreira2006,FerreiraFerreira2013}, through a more parsimonious generation of administrative redexes, this time the even more parsimonious translation $\om{(\cdot)}$ reaches a qualitative jump, since it gives a representation of $\ipc$ truly free from commuting conversions. See Fig.~\ref{fig:summary} for a visualization of our main result in the context of other translations of $\ipc$.

\begin{figure}[t]\caption{The main result of the present paper in context. Simulation of a commuting conversion $M\to_R N$ in $\ipc$. System $\f$ to the right of dashed line. System $\fat$ between the dashed and dotted lines. Comparison of map $\cm{(\cdot)}$ (proposed in \cite{FerreiraFerreira2009,FerreiraFerreira2013}), map $\am{(\cdot)}$ (proposed in \cite{EspiritoSantoFerreira2020}), map $\rpm{(\cdot)}$ (proposed in \cite{EspiritoSantoFerreira2021}), and map $\om{(\cdot)}$ (proposed here). Some reductions are ``administrative'', as argued in \cite{EspiritoSantoFerreira2020}. The reductions $\cm{N}\to^*Q$ and $\am{N}\to^* Q$ go in the wrong direction, causing $M\to N$ to be mapped to the $\beta$-equalities $\cm{M}=_{\beta}\cm{N}$ and $\am{M}=_{\beta}\am{N}$. The map $\rpm{(\cdot)}$ simulates $M\to N$ with the help of atomization reductions $\rho\varrho$ added to $\f$ in \cite{EspiritoSantoFerreira2021}. The map proposed here makes the commuting conversion $M\to N$ disappear: $\om{M}=\om{N}$.} \label{fig:summary}
\[
\resizebox{12cm}{!}{\xymatrix{
		& {} \ar@{--}[ddddddd]&&&&&&&&{}\ar@{..}[dddddd]&	\\
		M  \ar[dddd]_R && \cm{M}\ar@{->>}[rr]_{\beta(admin)} && \am{M} \ar@{->>}[rrd]_{\beta(admin)} \ar@/^2pc/@{<<->>}[rrrrdd]^{\beta\eta}&&&&&&\rpm{M}\ar[dddd]^{\beta\rho\varrho}\ar@/_2pc/[llllll]_{\rho\varrho}\\ 
		&&&&&& P \ar@{->>}[dd]^{\beta}\\
		&&&&&&&&*+[F]{\om{M}=\om{N}}\\
		&&&&&& Q\\
		N  && \cm{N}\ar@{->>}[rr]_{\beta(admin)} && \am{N} \ar@{->>}[rru]^{\beta(admin)} \ar@/_2pc/@{<<->>}[rrrruu]_{\beta\eta}&&&&&&\rpm{N}\ar@/^2pc/[llllll]^{\rho\varrho}
		\\
		\ipc && \fat &&&&&&&{}&\\
		&{}&&&&&&&&\f&
	}
}
\]
\end{figure}
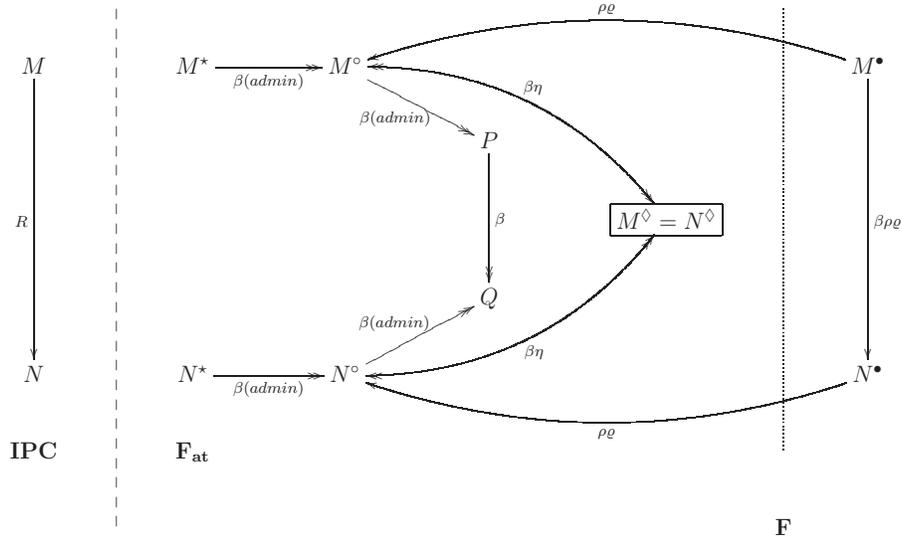

%

Technically, the optimization is possible because we can identify in the definition of the previous translation $\am{(\cdot)}$ the places responsible for the generation of administrative redexes, and program the translations slightly differently, so that it reduces ``at compile time'' any possible redex thus generated. In fact, this is just the initial idea: if applied naively, the resulting translation would reduce too much ``on the fly'' and identify much more than just the commuting conversions. A complementary idea is needed to protect source redexes from being unduly reduced - but this second idea and such technical details can be better explained and appreciated later on.

The paper is organized as follows. In Section \ref{sec:systems} we recall systems $\ipc$ and $\fat$. The new translation is introduced and discussed in Section \ref{sec:optimized}. In Section \ref{sec:simulation} we prove our main result, the simulation theorem. A part of the proof, concerning one case of commuting conversion, is given separately, in Section \ref{sec:special}. Section \ref{sec:discussion} discusses the simulation theorem, in particular identifies a subset of full $\beta\eta$-reduction of $\ipc$ that is strictly preserved by the translations proposed here. Section \ref{sec:final} concludes the paper.
\section{Systems}\label{sec:systems}

In this section the two systems used in the paper are recalled: the full intuitionistic proposicional calculus ($\ipc$) and the atomic polymorphic system ($\fat$), both presented in the (operational) $\lambda$-calculus style.   
\subsection{System $\ipc$}

Types/formulas are given by
$$
A,B,C\,::=\,X\,|\,\perp\,|\,A\supset B\,|\,A\wedge B\,|\,A\vee B
$$
where $X$ ranges over a denumerable set of \emph{type variables}. We define $\neg A:=A\supset\perp$.

Proof terms:
$$
{\small{
\begin{array}{rcll}
M,N,P,Q&::=&x&\textrm{(assumption)}\\
&|&\lb x^{A}.M\,|\,MN&\textrm{(implication)}\\
&|&\pair MN\,|\,\proj 1M\,|\,\proj 2M&\textrm{(conjunction)}\\
&|&\injn 1MAB\,|\,\injn 2NAB\,|\,\case M{x^A}P{y^B}QC&\textrm{(disjunction)}\\
&|&\abort MA&\textrm{(absurdity)}
\end{array}}}
$$
where $x$ ranges over a denumerable set of term (or assumption) variables. The type annotations in the proof terms $\injn iMAB$, $\case M{x^A}P{y^B}QC$ and $\abort MA$ are often omitted when no confusion arises. Also the pair $\pair{P_1}{P_2}$ is sometimes written as $\ipair i{P_i}$.

In types, all occurrences of type variables are free. In proof terms, there are occurrences of type variables, all of them free, and occurrences of term variables, which are free unless they belong to the scope of a binder for the occurring variable - in which case the occurrence is said to be bound. In $\lb x^A.M$, $\lb x^A$ is a binder of $x$ with scope $M$; and in $\case M{x^A}P{y^B}QC$ there is the binder $x^A$ of variable $x$ with scope $P$ and the binder $y^B$ of variable $y$ with scope $Q$. From these prescriptions, it is routine to define, by recursion on $P$, $\FV P$ (resp.~$\FTV P$), the set of term (resp.~type) variables with free occurrences in $P$. Often we write $x\in M$ (resp. $X\in M$) to mean $x\in \FV M$ (resp.~$X\in \FTV M$). We work modulo $\alpha$-equivalence, in particular we assume we can rename the bound term variables when necessary. 

The typing/inference rules are in Fig.~\ref{fig:typing}. $\Gamma$
denotes a set of \emph{declarations} $x:A$ such that a variable is
declared at most one time in $\Gamma$.
\begin{figure}[t]\caption{Typing/inference rules}\label{fig:typing}
$$
\begin{array}{c}
\infer[Ass]{\Gamma,x:A\vdash x:A}{}\\ \\
\infer[\supset I]{\Gamma\vdash\lb x^{A}.M:A\supset B}{\Gamma,x:A\vdash M:B}\qquad
\infer[\supset E]{\Gamma\vdash MN:B}{\Gamma\vdash M:A\supset B&\Gamma\vdash N:A}\\ \\
\infer[\wedge I]{\Gamma\vdash \pair MN:A\wedge B}{\Gamma\vdash
M:A&\Gamma\vdash N:B}\qquad
\infer[\wedge E1]{\Gamma\vdash \proj 1M:A}{\Gamma\vdash M:A\wedge B}\qquad\infer[\wedge E2]{\Gamma\vdash \proj 2M:B}{\Gamma\vdash M:A\wedge B}\\ \\
\infer[\vee I1]{\Gamma\vdash\injn 1MAB:A\vee B}{\Gamma\vdash M:A}\qquad
\infer[\vee I2]{\Gamma\vdash\injn 2NAB:A\vee B}{\Gamma\vdash N:B}\\ \\
\infer[\vee E]{\Gamma\vdash\case M{x^A}P{y^B}QC:C}{\Gamma\vdash
M:A\vee B&\Gamma,x:A\vdash P:C&\Gamma,y:B\vdash Q:C}\\ \\
\infer[\perp E]{\Gamma\vdash\abort MA:A}{\Gamma\vdash M:\perp}
\end{array}
$$
\end{figure}

The reduction rules of $\ipc$ are given in Fig.~\ref{fig:reduction-rules}. See \cite{EspiritoSantoFerreira2021} for more details on the reduction rules, including the ``subject reduction'' property, which connects reduction of proof terms with normalization of typing/inference derivations.


As usual, given a reduction rule $R$ on the proof terms, we denote by $\rightarrow_R$ the compatible closure of $R$, and then $\rightarrow_R^+$, $\rightarrow_R^*$ and $=_R$ denote respectively the transitive closure, the reflexive-transitive closure, and the reflexive-symmetric-transitive closure of $\rightarrow_R$. 
\begin{figure}\caption{Reduction rules of $\ipc$}\label{fig:reduction-rules}
	Detour conversion rules:
	$$
	\begin{array}{rrcll}
		(\betai)&(\lb x.M)N&\to&[N/x]M&\\
		(\betac)&\proj i{\pair{M_1}{M_2}}&\to&M_i&\textrm{ ($i=1,2$)}\\
		(\betad)&\cse{\injn iM{A_1}{A_2}}{x_1^{A_1}}{P_1}{x_2^{A_2}}{P_2}&\to&[M/x_i]P_i&\textrm{ ($i=1,2$)}
	\end{array}
	$$
	Commuting conversion rules for disjunction (in the 2nd rule, $i\in\{1,2\}$):
	$$
	{\small{\begin{array}{rrcl}
		(\pii) & (\case{M}{x}{P}{y}{Q}{C\supset D})N & \to & \case{M}{x}{PN}{y}{QN}{D}\\
		(\pic) & \proj i{(\case{M}{x}{P}{y}{Q}{C_1\wedge C_2})} & \to & \case{M}{x}{\proj iP}{y}{\proj iQ}{C_i}\\ 
		(\pid) & \\
		\multicolumn{2}{c}{\case{\case{M}{x'}{P'}{y'}{Q'}{C\vee D}}{x^C}{P}{y^D}{Q}{E}}  & \to & \\ 
		\multicolumn{4}{c}{\case M{x'}{\case{P'}{x^C}{P}{y^D}{Q}{E}}{y'}{\case{Q'}{x^C}{P}{y^D}{Q}{E}}{E}}\\
		(\pia) & \abort{\case{M}{x}{P}{y}{Q}{\perp}}{C} & \to &\\
		\multicolumn{4}{c} {\case{M}{x}{\abort P C}{y}{\abort Q C}{C}}
	\end{array}}}
	$$
	
	Commuting conversion rules for absurdity (in the 2nd rule, $i\in\{1,2\}$):
	$$
	\begin{array}{rrcl}
		(\abi) & (\abort M{C\supset D})N & \to & \abort MD\\
		(\abc) & (\abort M{C_1\wedge C_2})i & \to & \abort M{C_i}\\ 
		(\abd) & \case{\abort M{C \vee D}}{x^C}P{y^D}Q{E} & \to & \abort ME \\
		(\aba) & \abort{\abort M{\perp}}C & \to & \abort MC
	\end{array}
	$$
	$\eta$-rules:
	$$
	\begin{array}{rrcll}
		(\etai)&\lb x.Mx&\to&M&(x\notin M)\\
		(\etac)&\pair{\proj 1M}{\proj 2M}&\to&M&\\
		(\etad)&\case M{x^A}{\injn 1x{A}{B}}{y^B}{\injn 2yAB}{A\vee B}&\to&M&
	\end{array}
	$$
\end{figure}

We let
$$
\begin{array}{rcl}
\beta&:=&\betai\cup\betac\cup\betad \,\,\,\textrm{(and similarly for $\eta$)}\\
\pi&:=&\pii\cup\pic\cup\pid\cup\pia \,\,\,\textrm{(and similarly for $\varpi$)}\\
\end{array}
$$

\subsection{System $\fat$}

The atomic polymorphic system $\fat$ \cite{FerreiraFerreira2013} is a predicative fragment of system $\f$ in which all universal instantiations have an atomic witness. We define system $\fat$ by saying what changes relatively to $\ipc$.

Regarding types, $\perp$ and $A\vee B$ are dropped, and the new form $\forall X.A$ is adopted. In $\forall X.A$, $\forall X$ is a binder of variable $X$ with scope $A$. So type variables can now have bound occurrences in types.

Regarding proof-terms, the constructions relative to $\perp$ and $A\vee B$ are dropped, and the new forms $\Lambda X.M$ and $MX$ are adopted. In $\Lambda X.M$, $\Lambda X$ is a binder of variable $X$ with scope $M$. So type variables can now have bound occurrences in proof terms. 

Again, we will write $x\in \FV M$ and $X\in \FTV M$, often abbreviated $x\in M$, $X\in M$, respectively; and we can also define $\FTV A$, the set of type variables with free occurrences in $A$, and write $X\in \FTV A$, abbreviated $X\in A$. We always work modulo $\alpha$-equivalence, which encompasses the renaming of both type and term variables, both in terms and in types.


Regarding typing rules, those relative to $\perp$ and $A\vee B$ are dropped, and two rules relative to $\forall X.A$ are adopted:
$$
\infer[\forall I]{\Gamma\vdash \Lambda X.M:\forall X.A}{\Gamma\vdash M:A}\qquad\infer[\forall E]{\Gamma\vdash MY:[Y/X]A}{\Gamma\vdash M:\forall X.A}
$$
where the proviso for $\forall I$ is: $X$ occurs free in no type in $\Gamma$. 


Regarding reduction rules, we drop commuting conversions (since they are relative to $\vee$ and $\perp$). What remains are the $\beta$ and $\eta$-rules (but we drop those relative to disjunction). For $\forall$, these are:
$$
\begin{array}{rrcll}
(\betaall)&(\Lb X.M)Y&\to&[Y/X]M&\\
(\etaall)&\Lb X.MX&\to&M&(X\notin M)\\
\end{array}
$$

We let
$$
\begin{array}{rcl}
\beta&:=&\betai\cup\betac\cup\betaall \,\,\,\textrm{(and similarly for $\eta$)}\\
\end{array}
$$


\section{The optimized embedding}\label{sec:optimized}

In this paper we present a new translation $\om{(\cdot)}$ of $\ipc$ into $\fat$, a variant of the translation $\am{(\cdot)}$ introduced in \cite{EspiritoSantoFerreira2020}. Both maps employ the so-called Russell-Prawitz translation of formulas, so they do not differ at this level. At the proof-term level, they are both based on the idea of using admissibility in $\fat$ of the elimination of disjunction and absurdity (rather than the idea of instantiation overflow \cite{FerreiraFerreira2009,FerreiraFerreira2013}). They start to differ in the precise constructions that witness such admissibility. 

We now recognize that: (1) the constructions $\acasesymb$ and $\aabortsymb$ employed by $\am{(\cdot)}$ are the cause of some \emph{administrative redexes}, redexes which are not the translation of redexes already present in the proof being translated, but which are rather generated by the translation $\am{(\cdot)}$ itself; and (2) some of those administrative redexes are generated in the course of the translation of some commuting conversions of $\ipc$, and their reduction causes the failure of $\am{(\cdot)}$ in delivering a simulation of such conversions \cite{EspiritoSantoFerreira2020}. 

The starting point for the new translation $\om{(\cdot)}$ is to employ \emph{optimized} versions $\accasesymb$ and $\aabbortsymb$ where such defects are avoided, by forcing the reduction ``at compile time'' of such unwanted redexes. To this end, we introduce variants of the elimination constructions of $\fat$, denoted  $\pproj iM$, $M@N$ and $M@Y$. But the adoption of such elimination constructors cannot be restricted to the definition of the new constructors $\accasesymb$ and $\aabbortsymb$, it has to be effected in a more general way - to restore some coherence in the definition of $\om{(\cdot)}$ without which new defects would be introduced in the simulation. This will be discussed later, in the second subsection. Let us start with the technical development.
 

\subsection{The new translation}

The \emph{optimized translation} $\om{(\cdot)}$  of $\ipc$ into $\fat$ comprises the Russell-Prawitz translation of formulas and a translation of proof-terms (which induces a translation of derivations).

\begin{defn} In $\fat$:
	\begin{enumerate}
		\item $A\dvee B:=\forall X.((A\supset X)\wedge(B\supset X))\supset X$, with $X\notin A,B$.
		\item $\dperp:=\forall X.X$.
	\end{enumerate}
\end{defn}

We define the Russell-Prawitz translation of formulas. Using the abbreviations just introduced, the definition can be given in a homomorphic fashion:

$$
\begin{array}{rcl}
	\om X&=&X\\
	\om \perp&=&\dperp\\
	\om{(A\supset B)}&=&\om A\supset\om B\\
	\om{(A\wedge B)}&=&\om A\wedge\om B\\
	\om{(A\vee B)}&=&\om A\dvee \om B
\end{array}
$$

The translation of proof terms will rely on the following definitions:

\begin{defn} In $\fat$:
	\begin{enumerate}
		\item $\llb x.M :=\lambda w.(\lambda x M)w$, with $w\notin M$.  
		\item $\ppair{M}{N}:=\pair{\proj1{\pair{M}{N}}}{\proj2{\pair{M}{N}}}$.
		\item $\pproj iM$, $M@N$ and $M@Y$ are defined in Fig.~\ref{fig:optimized-elims}.
	\end{enumerate}
\end{defn}
Iterated applications of $@$ are bracketed to the right; e.g. $M@x@X$ denotes $(M@x)@X$.

\begin{figure}\caption{Optimized elimination constructions}\label{fig:optimized-elims}
	$$
	\begin{array}{c}
		\pproj iM:=\left\{\begin{array}{ll}N_i&\textrm{if $M=\pair{N_1}{N_2}$}\\\proj iM&\textrm{otherwise}\end{array}\right.
		\\ \\
		M@N:=\left\{\begin{array}{ll}[N/x]P&\textrm{if $M=\lb x.P$}\\MN&\textrm{otherwise}\end{array}\right.
		\\ \\
		M@Y:=\left\{\begin{array}{ll}[Y/X]P&\textrm{if $M=\Lb X.P$}\\MY&\textrm{otherwise}\end{array}\right.
	\end{array}
	$$
\end{figure}

So we not only introduced optimized variants of $\proj{i}{M}$, $MN$ and $MY$, but also defined expanded variants of $\lb$-abstraction and pair. Next we define the constructors which will witness the inference rules of disjunction and absurdity.

\begin{defn}\label{def:admissible-constructions} In $\fat$:
	\begin{enumerate}
		\item Given $M,A,B$, given $i\in\{1,2\}$, we define
		$$\aiinjn iMAB:=\Lb X.\lb w^{(A\supset X)\wedge(B\supset X)}.\proj i w M,$$
		where the bound variable $X$ (resp.~$w$) is chosen so that $X\notin M,A,B$ (resp.~$w\notin M$).
		\item Given $M,P,Q,A,B,C$, we define $\accase M{x^A}P{y^B}QC$ by recursion on $C$ as follows:
	
		$$
		\begin{array}{rcl}
			\accase MxPyQX&=&(M@X)@\pair{\lb x.P}{\lb y.Q}\\
			\accase M{x}{P}{y}{Q}{C_1\wedge C_2}&=&\ipair i{\accase M{x}{\pproj i{P}}{y}{\pproj i{Q}}{C_i}}\\
			\accase MxPyQ{C\supset D}&=&\lb z.\accase Mx{P@z}y{Q@z}{D}\\
			\accase M{x}P{y}Q{\forall X.C}&=&\Lb X.\accase Mx{P@X}y{Q@X}C
		\end{array}
		$$
		where, in the third clause, the bound variable $z$ is chosen so that $z\neq x$, $z\neq y$ and $z\notin M,P,Q$; and in the fourth clause, the bound variable $X$ is chosen so that $X\notin M,P,Q,A,B$.
		\item Given $M,A$, we define $\aabbort MA$ by recursion on $A$ as follows:
	$$
	\begin{array}{rcl}
		\aabbort MX &=& M@X\enspace\\
		\aabbort M{A_1\wedge A_2}&=&\pair{\aabbort M{A_1}}{\aabbort M{A_2}}\\
		\aabbort M{B\supset C}&=&\lb z^B.\aabbort M{C}\\
		\aabbort M{\forall X.A}&=&\Lb X.\aabbort MA
	\end{array}
	$$
		where, in the third clause, the bound variable $z$ is chosen so that $z\notin M$; and in the fourth clause, the bound variable $X$ is chosen so that $X\notin M$.
		
	\end{enumerate}
\end{defn}

The definition of $\aiinjsymb$ is exactly the same as the definition of $\ainjsymb$ employed in $\am{(\cdot)}$. The difference between the original $\acasesymb$ and $\aabortsymb$ employed in $\am{(\cdot)}$ and $\accasesymb$ and $\aabbortsymb$ to be employed in $\om{(\cdot)}$ is in the use of $@$: if $@$ is replaced by application and projection, we obtain the original constructors.

We are ready to give the translation of proof terms.
\begin{defn}[Optimized translation]
Given $M\in\ipc$, $\om M$ is defined by recursion on $M$ as in Fig.~\ref{fig:optimized-translation}.
\end{defn}
We also define $\om{\Gamma}:=\{(x:\om{A}) \mid (x:A)\in\Gamma\}$.

\begin{figure}[t]\caption{The optimized translation of proof terms}\label{fig:optimized-translation}
$$
\begin{array}{rcll}
\om x&=&x&\\
\om{(\lb x^A.M)}&=&\llb x^{\om A}.\om M&\\
\om{\pair MN}&=&\ppair{\om M}{\om N}&\\
\om{(\injn iMAB)}&=&\aiinjn i{\om M}{\om A}{\om B}&\textrm{($i=1,2$)}\\
\om{(MN)}&=&(\om M)@(\om N)&\\
\om{(\proj iM)}&=&\pproj i{(\om M)}&\textrm{($i=1,2$)}\\
\om{(\case M{x^A}P{y^B}QC)}&=&\accase{\om M}{x^{\om A}}{\om P}{y^{\om B}}{\om Q}{\om C}&\textrm{($P,Q:C$)}\\
\om{(\abort MA)}&=&\aabbort{\om M}{\om A}\\
\end{array}
$$
\end{figure}

\begin{lem}\label{lem:typing-optimized-constructions}In $\fat$:
\begin{enumerate}
\item The typing rules for $\lb x.M$, $\pair MN$, $MN$, $\proj iM$ and $MX$ 
also hold, for the expansion/optimized variants $\llb x.M$, $\ppair MN$,  $M@N$, $\pproj iM$ and $M@X$ respectively.
\item The typing rules in Fig.\ref{fig:admissible-typing-rules} are admissible in $\fat$. 
\end{enumerate}
\end{lem}
\begin{proof}
1. In the case of the expansion variants $\llb x.M$, $\ppair MN$ the result follows immediately, as derived rules. In the case of the optimized variants $M@N$, $\pproj iM$ and $M@X$ the result follows from the fact that the typing rules for $\pair{N_1}{N_2}$, $\lb x.P$, and $\Lb X.P$ are invertible and due to the typing rules for $[N/x]P$ and $[Y/X]P$.

2. From \cite{EspiritoSantoFerreira2020} we know that the typing rules of Fig.\ref{fig:admissible-typing-rules} for some constructions there introduced, called $\ainjn iM{A_1}{A_2}$, $\acase MxPyQC$ and $\aabort MC$ are admissible in $\fat$.  The result is then immediate since $\aiinjn iM{A_1}{A_2}$ coincide with $\ainjn iM{A_1}{A_2}$ and the variants $\accase MxPyQC$ and $\aabbort MC$ are defined from $\acase MxPyQC$ and $\aabort MC$ by replacing elimination occurrences by their optimized variants and the optimized variants satisfy the same typing rules (by 1.). \end{proof}
\begin{figure}[t]\caption{Admissible typing rules of $\fat$}\label{fig:admissible-typing-rules}
	$$
	\begin{array}{c}
	\infer[(i=1,2)]{\Gamma\vdash\aiinjn iM{A_1}{A_2}:A_1\dvee A_2}{\Gamma\vdash M:A_i}\\ \\
	\infer{\Gamma\vdash\accase M{x^A}P{y^B}QC:C}{\Gamma\vdash M:A\dvee B&\Gamma,x:A\vdash P:C&\Gamma,y:B\vdash Q:C}\\ \\
		\infer{\Gamma\vdash\aabbort MC:C}{\Gamma\vdash M:\perp}
	\end{array}
	$$
\end{figure}
The following proposition is the soundness of the translation. Although simple, its proof is important because implicitly defines the translation of derivations induced by the proof-term translation.
\begin{prop}
	If $\Gamma\vdash M:A$ in $\ipc$ then $\om{\Gamma}\vdash\om{M}:\om{A}$ in $\fat$.
\end{prop}
\begin{proof}
	By induction on $\Gamma\vdash M:A$, using Lemma \ref{lem:typing-optimized-constructions}.
\end{proof}


We end this subsection with a very brief comparison between the old translation $\am{(\cdot)}$ and the new $\om{(\cdot)}$. At the level of proof-term translation, the novelties of $\om{(\cdot)}$ w.~r.~t.~$\am{(\cdot)}$ are as follows: not only we employ the new $\accasesymb$ and $\aabbortsymb$, but also we employ the expanded $\lb$-abstraction and pair; and $@$ is not restricted to definition of $\accasesymb$ and $\aabbortsymb$, it is used in the translation of application and projection.

The connection between the two translations follows from some immediate observations about the new constructors.

\begin{lem}\label{lem:optimization}In $\fat$:
\begin{enumerate}
\item $Mi\to_{\betac}^= M@i$.
\item $MN\to_{\betai}^= M@N$.
\item $MX\to_{\betaall}^= M@X$.
\end{enumerate}
\end{lem}
\begin{proof}
Immediate.
\end{proof}

The previous lemma says that, by employing $@$ instead of application or projection, some $\beta$-reduction may be forced.

\begin{lem}\label{lem:admissible-eta-reduction}In $\fat$:
\begin{enumerate}
\item $\llb x. M\to_{\etai}\lb x.M $.
\item $\ppair{M}{N}\to_{\etac} \pair{M}{N}$.
\item $(\llb x.M)@N=(\lb x.M)N\to_{\betai}(\lb x.M)@N$.
\item $\ppair{M_1}{M_2}@i= \pair{M_1}{M_2}i \to_{\betac}\pair{M_1}{M_2}@i$.
\item $\lb z.M@z\to_{\etai}^=M$, if $z\notin M$.
\item $\pair{\pproj 1M}{\pproj 2M}\to_{\etac}^=M$.
\item $\Lambda Y.M@Y\to_{\etaall}^=M$, if $Y\notin M$.
\end{enumerate}
\end{lem}
\begin{proof}
The proof of the first four items is immediate, the other three are proved by straightforward case analysis of $M$.
\end{proof}

Some comments about the previous lemma: The first two items confirm that $\llb$ and $\ppair{\_}{\_}$ are $\eta$-expanded versions of the original constructors. The third and fourth items say that, by employing $\llb$ and $\ppair{\_}{\_}$, we block the reduction that $@$ wants to force. The last three items say that $@$ may also force some $\eta$-reduction.


From these simple observations, it follows that $\om{M}$ is obtained from $\am{M}$ through some $\beta$-reduction steps (due to Lemma \ref{lem:optimization}) and some $\eta$-expansion steps (due to the first two items of Lemma \ref{lem:admissible-eta-reduction}). This justifies the links depicted in Fig.~\ref{fig:summary} between $\am{M}$ and $\om{M}$, and between $\am{N}$ and $\om{N}$.

\subsection{Discussion of the translation}

As hinted before, we want to change translation $\am{(\cdot)}$ because it does not simulate commuting conversions $\pid$ and $\pia$. As detailed in \cite{EspiritoSantoFerreira2020}, Lemmas 11 and 12, the failure is due to the fact that, if we define $\acasesymb$ and $\aabortsymb$ in the way employed by $\am{(\cdot)}$, the rules $\pid$ and $\pia$ are admissible in $\fat$ only in the form of $=_{\beta}$, requiring $\beta$-reduction from the \emph{contractum} of the rule. 

Inspecting the proofs of the mentioned lemmas, one sees that the redexes contracted in such unwanted reductions have to do with the clauses in the definition of $\acase{M}{x_1}{P_1}{x_2}{P_2}{C}$ relative to the cases where $C$ is not atomic. For instance, $\acase{M}{x_1}{P_1}{x_2}{P_2}{D\supset E}=\lb z^D.\acase{M}{x_1}{P_1 z}{x_2}{P_2 z}{E}$. The subterm $P_i z$ turns out to be an administrative redex, as soon as $P_i$ is a $\lb$-abstraction - for instance when $P_i$ is the translation of a $\casesymb$, which is exactly what happens in the proofs of the mentioned lemmas. By adopting $P_i@z$ in the definition of $\accasesymb$, such redex is immediately reduced.

As it happens, we cannot simply add $@$ in those clauses of $\accasesymb$. If we want to reprove the new lemmas corresponding to those mentioned lemmas of \cite{EspiritoSantoFerreira2020}, the case of $C$ atomic breaks, unless $@$ is also employed in the base cases of the definition of $\accasesymb$ and $\aabbortsymb$. By doing this further change, the new lemmas work unexpectedly well: $\pid$ and $\pia$ become admissible in $\fat$ as syntactic identities (see Lemmas \ref{lem:admissible-pi-disjunction-reduction} and \ref{lem:admissible-pi-absurdity-reduction} below). In fact, all commuting conversions are collapsed in this way, except for $\pii$, because we would have $(\accase MxPyQ{C\supset D})N\to_{\beta}^*\accase Mx{P@N}y{Q@N}D$, while the \emph{contractum} $\accase Mx{PN}y{QN}D$ would possibly require reduction (in the ``wrong direction'') to bring $PN$ and $QN$ to the forms $P@N$ and $Q@N$ respectively (recall item 2 of Lemma \ref{lem:optimization}). 

The solution was to use $@$ also in the translation of $MN$ (and $\proj{i}{M}$), so that the reduction $(\accase MxPyQ{C\supset D})@N\to_{\beta}^*\accase Mx{P@N}y{Q@N}D$ in $\fat$ is the obtained admissibility of $\pii$ (see Lemma \ref{lem:optimized-admissible-commutative-conversions} below) - \emph{h\'elas} not as syntactic identity, contrary to the other commuting conversions. It turns out that, now taking $M$, $P$, $Q$, $C$, $D$ in $\ipc$, one has
$$(\accase {\om M}x{\om P}y{\om Q}{\om C\supset \om D})@N=\accase {\om M}x{\om P @ \om N}y{\om Q@ \om N}{\om D}$$
but to establish this requires an entire section of the paper - see Section \ref{sec:special}.

Finally, the generalized adoption of $@$ in the translation implies that also the $\beta$-redexes present in the proof term to be translated are reduced ``on the fly'' - unless such reduction is blocked, as shown in items 3 and 4 of Lemma \ref{lem:admissible-eta-reduction}, by the trick of translating source $\lb$-abstractions and pair in a $\eta$-expanded way. In this way, a large portion of source $\beta$-reduction is strictly simulated, as made precise at the end of Section \ref{sec:simulation}, and this explains the last design decision incorporated in the definition of $\om{(\cdot)}$.

\section{The simulation theorem}\label{sec:simulation}

In this section we prove our main result, the simulation property of the optimized embedding. The proof only comes in the third subsection, after several preliminary results about $\fat$ are established in the first and second subsections.

\subsection{Preliminary results}

In this subsection we prove bureaucratic results about the optimized constructors $\llb$, $\ppair{\_}{\_}$, $@$, $\accasesymb$ and $\aabbortsymb$, concerning: free variables (Lemmas \ref{lem:FV_om_@}, \ref{lem:FV_case}, \ref{lem:miracle_abort_FV}), compatibility of reduction (Lemma \ref{lem:additional-compatibility-rules-1}, Corollary \ref{cor-lem:additional-compatibility-rules-1}, Lemma  \ref{lem:additional-compatibility-rules-2}), and substitution (Lemmas \ref{lem:term-subst-@}, \ref{lem:type-subst-@}, \ref{lem:term-subst-expansions}, \ref{lem:subst-admissible-optimized-constructions}). 

\begin{lem}\label{lem:FV_om_@}
	In $\fat$, let $M$ be a term and $U$ be a term or a type variable. Then $\FV M\subseteq \FV{M@U}$. 
	Moreover, in the case of $U$ being a type variable, $\FV{M}=\FV{M@U}$.
\end{lem}

\begin{proof}
	In the case $M$ is not an abstraction, the result is clear since $M@U= MU.$ If $M$ is a lambda-abstraction (respectively a universal-abstraction) we only need to analyse the cases in which $U$ is a term (respectively a type variable). In each case the substitution ``on the fly'' promoted by $@$ maintain the free variables as free variables. The exact equality in the case of type variables is justified by the fact that application with a type variable or substitution of a type variable by another one does not add new 
	term variables.
\end{proof}

\noindent Obviously, $\FV{M@i}$ does not necessarily contain $\FV M$.

We now consider $\accasesymb$ and $\aabbortsymb$. Notice that, even if $w\in \FV{P}$ or $w\in \FV{Q}$, we have no guarantee that $w\in \FV{\accase MxPyQC}$ - consider the case $C=X$ and $M=\Lambda Y.\lb z.N$ with $z\notin \FV{N}$. On the other hand, the free variables of $M$ are not lost:

\begin{lem}\label{lem:FV_case}
	In $\fat$, $\FV{M}\subseteq \FV{\accase MxPyQC}$.
	
\end{lem}

\begin{proof}
	The proof is by induction on $C$.
	
	Case $C=X$. $\FV{\accase MxPyQC}=\FV{M@X@\pair{\lb x. P}{\lb y.Q}}$. Applying Lemma \ref{lem:FV_om_@} twice, we obtain  $$\FV{M}=\FV{M@X}\subseteq \FV{M@X@\pair{\lb x. P}{\lb y.Q}}.$$
	
	Case $C=C_1\wedge C_2$. Let $j\in\{1,2\}$.
	
	$$
	\begin{array}{rcll}
		\FV{M} & \subseteq & \FV{\accase Mx{P@j}y{Q@j}{C_j}} & \text{(by IH)}\\
		& \subseteq & \bigcup_{i=1,2}\FV{\accase Mx{P@i}y{Q@i}{C_i}}\\
		&=&\FV{\ipair i{\accase Mx{P@i}y{Q@i}{C_i}}}&\\
		&=&\FV{\accase MxPyQC}&\textrm{(by def. of $\accasesymb$)}
	\end{array}
	$$
	%
	
	Case $C=C_1\supset C_2$.
	
	$$
	\begin{array}{rcll}
		\FV{M} & \subseteq & (\FV{\accase Mx{P@z}y{Q@z}{C_2}})\setminus z&\text{(*)}\\
		&=&\FV{\lb z^{C_1}.\accase Mx{P@z}y{Q@z}{C_2}}&\\
		&=&\FV{\accase MxPyQC}& \textrm{(by def. of $\accasesymb$)}
	\end{array}
	$$
	Justification $(*)$. Let $w\in\FV{M}$. By IH, $w\in \FV{\accase Mx{P@z}y{Q@z}{C_2}}$. 
	We may assume $w\neq z$. Thus $w\in (\FV{\accase Mx{P@z}y{Q@z}{C_2}})\setminus z$. 
	
	Case $C=\forall X C_0$.
	
	$$
	\begin{array}{rcll}
		\FV{M} &\subseteq& \FV{\accase Mx{P@X}y{Q@X}{C_0}} & \text{(by IH)} \\
		&=&\FV{\Lb X.\accase Mx{P@X}y{Q@X}{C_0}}&\\
		&=& \FV{\accase MxPyQC} & \text{(by def. of $\accasesymb$)}
	\end{array}
	$$
	%
\end{proof}
\begin{lem}\label{lem:miracle_abort_FV}
	In $\fat$, $\FV{\aabbort{M}{A}}=\FV M$.
\end{lem}

\begin{proof}
	The proof is by induction on $A$.
	
	Case $A=X$. 
	
	$\aabbort MA=M@X$. If $M$ is not an universal lambda abstraction then $M@X=MX$ and $\FV {MX}=\FV M$. If $M=\Lb Y. M_0$ then $M@X=[X/Y]M_0$ and $\FV{[X/Y]M_0}=\FV{M_0}=\FV{\Lb Y.M_0}=\FV M$.
	
	Case $A=A_1\wedge A_2$.
	
	$$
	\begin{array}{rcll}
		\FV{\aabbort M{A_1\wedge A_2}}&=&\FV{\ipair i{\aabbort M{A_i}}}&\textrm{(by def. of $\aabbortsymb$)}\\
		&=&\bigcup_{i=1,2}\FV{\aabbort M {A_i}}&\\
		&=&\bigcup_{i=1,2} \FV M&\textrm{(by IH twice)}\\
		&=&\FV{M}&
	\end{array}
	$$
	
	Case $A=A_1\supset A_2$.
	
	$$
	\begin{array}{rcll}
		\FV{\aabbort M{A_1\supset A_2}}&=&\FV{\lb z^{A_1}.\aabbort M{A_2}}&\textrm{(by def. of $\aabbortsymb$)}\\
		&=&\FV{\aabbort M{A_2}}\setminus z&\\
		&=&\FV{M}\setminus z&\textrm{(by IH)}\\
		&=&\FV{M}& \textrm{(since $z\notin \FV{M}$)}
	\end{array}
	$$
	
	Case $A=\forall X. A_0$.
	
	$$
	\begin{array}{rcll}
		\FV{\aabbort M{\forall X. A_0}}&=&\FV{\Lb X.\aabbort M{A_0}}&\textrm{(by def. of $\aabbortsymb$)}\\
		&=&\FV{\aabbort M{A_0}}&\\
		&=&\FV{M}&\textrm{(by IH)}
	\end{array}
	$$
	
\end{proof}

\begin{lem}\label{lem:additional-compatibility-rules-1}Let $R$ be a redunction rule of $\fat$. Then 
	\begin{enumerate}
		\item \begin{enumerate}
			\item If $M \to_{R}M'$ then $M@i \to_{R\betac}^kM'@i$, with $k\in \{0,1,2\}$.
			\item If $M \to_{R}M'$ then $M@X \to_{R\betaall}^kM'@X$, where: if $R\neq \etaall$ then $k\in \{1,2\}$, else $k\in \{0,1,2\}$.
			\item If $M \to_{R}M'$ then $M@N \to_{R\betai}^kM'@N$, where: if $R\neq \etai$ then $k\in \{1,2\}$, else $k\in \{0,1,2\}$.
		\end{enumerate}
		\item If $N \to_{R}N'$ then $M@N \to_{R}^*M@N'$.
		\item If $M \to_{R}M'$ then $\ppair{M}{N} \to_{R}^2\ppair{M'}{N}$.
		\item If $N \to_{R}N'$ then $\ppair{M}{N} \to_{R}^2\ppair{M}{N'}$.
		\item If $M \to_{R}M'$ then $\llb x.M\to_{R}\llb.x M'$.
	\end{enumerate}
\end{lem}
\begin{proof}
	Note that, $R$ being a reduction on $\fat$, $\to_{R}$ is a compatible relation on the proof terms of $\fat$.
	
	1.(a) Suppose that $M\to_{R}M'$. 
	There are two cases.
	
	If $M$ is not of the form $\pair{N_1}{N_2}$, then $M@i=Mi \to_R M'i \to_{\betac}^{=}M'@i$, where the last step uses Lemma \ref{lem:optimization}.
	
	If $M=\pair{N_1}{N_2}$, and since $M\to_{R}M'$, there are two sub-cases. 
	
	In the first sub-case, we have either $M'=\pair{N'_1}{N_2}$ and $N_1\to_{R}N'_1$, or $M'=\pair{N_1}{N'_2}$ and $N_2\to_{R}N'_2$. But then, if the reduction happens in the $i$-th component of the pair,
	$M@i=N_i\to_{R} N'_i =M'@i$; otherwise $M@i=N_i=M'@i$.
	
	In the second sub-case, we have $M=\pair{M'1}{M'2}\to _{\etac}M'$ (hence $R=\etac$). By definition of $@$ and Lemma \ref{lem:optimization}, we obtain $M@i=M'i \to_{\betac}^= M'@i$.
	
	1.(b) Suppose that $M\to_{R}M'$. There are two cases.
	
	If $M$ is not of the form $\Lambda Y.P$, then $M@X=MX \to_R M'X \to_{\betaall}^{=}M'@X$, where the last step uses Lemma \ref{lem:optimization}.
	
	If $M=\Lambda Y.P$, and since $M\to_{R}M'$, there are two sub-cases:
	
	In the first sub-case, we have $M'=\Lambda Y.P'$ and $P\to_{R}P'$. But then $M@X=[X/Y]P\to_{R}[X/Y]P' =M'@X$, due to the fact that from $P\to_RP'$ we have $[X/Y]P\to_R[X/Y]P'$.
	
	In the second sub-case, we have $M=\Lambda Y.M'Y\to_{\etaall}M'$ (hence $R=\etaall$). By definition of $@$, by the fact that $Y\notin M'$ and Lemma \ref{lem:optimization}, we obtain $M@X=[X/Y](M'Y)=M'X \to_{\betaall}^= M'@X$.
	
	1.(c) Suppose that $M\to_{R}M'$. There are two cases.
	
	If $M$ is not of the form $\Lambda x.P$, then $M@N=MN \to_R M'N \to_{\betai}^{=}M'@N$, where the last step uses Lemma \ref{lem:optimization}.
	
	If $M=\Lambda x.P$, and since $M\to_{R}M'$, there are two sub-cases.
	
	In the first sub-case, we have $M'=\Lambda x.P'$ and $P\to_{R}P'$. But then 
	$M@N=[N/x]P\to_{R}[N/x]P' =(\lb x.P')@N=M'@N$, due to the fact that from $P\to_RP'$ we have $[N/x]P\to_R[N/x]P'$.
	
	In the second sub-case, we have $M=\Lambda x.M'x\to_{\etai}M'$ (hence $R=\etai$). By definition of $@$, by the fact that $x\notin M'$ and Lemma \ref{lem:optimization}, we obtain $M@N=[N/x](M'x)=M'N \to_{\betai}^= M'@N$.
	
	2. Suppose that $N\to_{R}N'$. Let us prove that $M@N\to_R^* M@N'$. 
	
	If $M$ is not $\lb x.P$, then $M@N=MN \to_R MN' =M@N'$.
	
	If $M$ is $\lb x.P$ then $M@N=[N/x]P\to_{R}^*[N'/x]P =(\lb x.P)@N'=M@N'$, due to the fact that from $N\to_RN'$ we have $[N/x]P\to_R[N'/x]P$.
	
	3. Suppose that $M\to_{R}M'$. Let us prove that $\ppair{M}{N}\to_R^2\ppair{M'}{N}$. 
	
	$$\begin{array}{rcll}
		\ppair{M}{N}&=&\pair{\pair{M}{N}1}{\pair{M}{N}2}&\textrm{(by def. of $\ppair{.}{.}$)}\\
		&\to_{R}^2&\pair{\pair{M'}{N}1}{\pair{M'}{N}2}&\textrm{(by compatibility)}\\
		&=&	\ppair{M'}{N}&\textrm{(by def. of $\ppair{.}{.}$)}
	\end{array}$$
	
	4. The proof is entirely similar to the previous item.
	
	5.  Suppose that $M\to_{R}M'$. Let us prove that $\llb x.M \to_{R} \llb x.M'$.
	
	$$\begin{array}{rcll}
		\llb x.M&=&\lb w.(\lb x.M)w&\textrm{(by def. of $\llb$)}\\
		&\to_{R}&\lb w.(\lb x.M')w&\textrm{(by compatibility)}\\
		&=&	\llb x.M'&\textrm{(by def. of $\llb$)}
	\end{array}$$

\end{proof}

By an easy inspection of Lemma \ref{lem:additional-compatibility-rules-1}, we extract the following result: 

\begin{cor}\label{cor-lem:additional-compatibility-rules-1}In $\fat$, let $\to^*$ (resp.~$\to^+$) denote $\to_{\beta\eta}^*$ (resp.~$\to_{\beta\eta}^+$).
	\begin{enumerate}	\item \begin{enumerate}
			\item If $M \to^+M'$ then $M@X \to^*M'@X$, and $M@N \to^*M'@N$, and $M@i \to^*M'@i$.
			\item If $M \to_{\beta}^+M'$ then $M@X \to_{\beta}^+M'@X$, and $M@N \to_{\beta}^+M'@N$, but $M@i \to_{\beta}^*M'@i$.
		\end{enumerate}
		\item \begin{enumerate}
			\item If $N \to^+N'$ then $M@N \to^*M@N'$.
			\item If $N \to_{\beta}^+N'$ then $M@N \to_{\beta}^*M@N'$.
		\end{enumerate}
		\item  \begin{enumerate}
			\item If $M \to^+M'$ then $\ppair{M}{N} \to^+\ppair{M'}{N}$.
			\item If $M \to_{\beta}^+M'$ then $\ppair{M}{N} \to_{\beta}^+\ppair{M'}{N}$.
		\end{enumerate}
		\item  \begin{enumerate}
			\item If $N \to^+N'$ then $\ppair{M}{N} \to^+\ppair{M}{N'}$.
			\item If $N \to_{\beta}^+N'$ then $\ppair{M}{N} \to_{\beta}^+\ppair{M}{N'}$.
		\end{enumerate}
		\item  \begin{enumerate}
			\item If $M \to^+M'$ then $\llb x. M \to^+\llb x. M'$.
			\item If $M \to_{\beta}^+M'$ then $\llb x. M \to_{\beta}^+\llb x. M'$.
		\end{enumerate}
	\end{enumerate}
\end{cor}


\begin{lem}\label{lem:additional-compatibility-rules-2}Let $R$ be a reduction rule of $\fat$. Then
	\begin{enumerate}
		\item If $M\to_RM'$ then $\aiinjn iMAB \to_R \aiinjn i{M'}AB$.
		\item If $M\to_RM'$ then $\accase MxPyQC \to_{R\beta}^* \accase {M'}xPyQC$. Moreover, if $M\to_{\beta}M'$ then \\$\accase MxPyQC \to_{\beta}^+ \accase {M'}xPyQC$.
		\item If $P\to^*P'$ then $\accase MxPyQC \to^* \accase {M}x{P'}yQC$. Moreover, if $P\to_{\beta}^*P'$ then $\accase MxPyQC \to_{\beta}^* \accase {M}x{P'}yQC$.
		\item If $Q\to^*Q'$ then $\accase MxPyQC \to^* \accase {M}x{P}y{Q'}C$. Moreover, if $Q\to_{\beta}^*Q'$ then $\accase MxPyQC \to_{\beta}^* \accase {M}x{P}y{Q'}C$.
		\item If $M\to_RM'$ then $\aabbort MA \to_{R\beta}^* \aabbort {M'}A$. Moreover, \\if $M\to_{\beta}M'$ then $\aabbort MA \to_{\beta}^+ \aabbort {M'}A$.
	\end{enumerate}	
\end{lem}
\begin{proof} 
	1. Suppose $M\to_RM'$. Then
	$$LHS=\Lambda X.\lb w.wiM\to_R\Lb X.\lb w.wiM'=RHS.$$
	
	2. Suppose $M\to_RM'$. By induction on $C$.
	
	Case $C=X$.
	$$
	\begin{array}{rcll}
		LHS&=&M@X@\pair{\lb x.P}{\lb y.Q}&\textrm{(by def. of $\accasesymb$)}\\
		&\to_{R\beta}^{\leq 4}&M'@X@\pair{\lb x.P}{\lb y.Q}&\textrm{(by Lemma \ref{lem:additional-compatibility-rules-1}.1)}\\
		&=&RHS&\textrm{(by def. of $\accasesymb$)}
	\end{array}
	$$
	
	Case $C=C_1\supset C_2$.
	$$
	\begin{array}{rcll}
		LHS&=&\lb z.\accase Mx{P@z}y{Q@z}{C_2}&\textrm{(by def. $\accasesymb$)}\\
		&\to_{R\beta}^*&\lb z.\accase {M'}x{P@z}y{Q@z}{C_2}&\textrm{(by IH)}\\
		&=&RHS&\textrm{(by def. $\accasesymb$)}
	\end{array}
	$$
	Case $C=C_1\wedge C_2$.
	$$
	\begin{array}{rcll}
		LHS&=&\ipair{i}{\accase Mx{\pproj iP}y{\pproj iQ}{C_i}}&\textrm{(by def. $\accasesymb$)}\\
		&\to_{R\beta}^*&\ipair{i}{\accase {M'}x{\pproj iP}y{\pproj iQ}{C_i}}&\textrm{(by IH twice)}\\
		&=&RHS&\textrm{(by def. $\accasesymb$)}
	\end{array}
	$$
	Case $C=\forall X. D$.
	$$
	\begin{array}{rcll}
		LHS&=&\Lb X.\accase Mx{P@X}y{Q@X}{D}&\textrm{(by def. $\accasesymb$)}\\
		&\to_{R\beta}^*&\Lb X.\accase {M'}x{P@X}y{Q@X}{D}&\textrm{(by IH)}\\
		&=&RHS&\textrm{(by def. $\accasesymb$)}
	\end{array}
	$$
	
	The more refined result when $R$ is a $\beta$-reduction comes from Lemma \ref{lem:additional-compatibility-rules-1}.1.(b) and (c).
	
	3. By induction on $C$.
	
	Case $C=X$.
	$$
	\begin{array}{rcll}
		LHS&=&M@X@\pair{\lb x.P}{\lb y.Q}&\textrm{(by def. of $\accasesymb$)}\\
		&\to^*&M@X@\pair{\lb x.P'}{\lb y.Q}&\textrm{(by Corollary \ref{cor-lem:additional-compatibility-rules-1}, since $\lb x.P\to^*\lb x.P'$)}\\
		&=&RHS&\textrm{(by def. of $\accasesymb$)}
	\end{array}
	$$
	
	Case $C=C_1\supset C_2$.
	$$
	\begin{array}{rcll}
		LHS&=&\lb z.\accase Mx{P@z}y{Q@z}{C_2}&\textrm{(by def. $\accasesymb$)}\\
		&\to^*&\lb z.\accase {M}x{P'@z}y{Q@z}{C_2}&\textrm{(by IH and Corollary \ref{cor-lem:additional-compatibility-rules-1})}\\
		&=&RHS&\textrm{(by def. $\accasesymb$)}
	\end{array}
	$$
	Case $C=C_1\wedge C_2$.
	$$
	\begin{array}{rcll}
		LHS&=&\ipair{i}{\accase Mx{\pproj iP}y{\pproj iQ}{C_i}}&\textrm{(by def. $\accasesymb$)}\\
		&\to^*&\ipair{i}{\accase {M}x{\pproj i{P'}}y{\pproj iQ}{C_i}}&\textrm{(by IH twice and Corollary \ref{cor-lem:additional-compatibility-rules-1})}\\
		&=&RHS&\textrm{(by def. $\accasesymb$)}
	\end{array}
	$$
	Case $C=\forall X. D$.
	$$
	\begin{array}{rcll}
		LHS&=&\Lb X.\accase Mx{P@X}y{Q@X}{D}&\textrm{(by def. $\accasesymb$)}\\
		&\to^*&\Lb X.\accase {M}x{P'@X}y{Q@X}{D}&\textrm{(by IH and Corollary \ref{cor-lem:additional-compatibility-rules-1})}\\
		&=&RHS&\textrm{(by def. $\accasesymb$)}
	\end{array}
	$$
	
	The more refined result for $\beta$-reductions comes from clauses (b) of Corollary \ref{cor-lem:additional-compatibility-rules-1}.
	
	4. Analogous to 3.
	
	5. Suppose $M\to_{R}M'$. By induction on $C$.
	
	Case $C=X$.
	$$
	\begin{array}{rcll}
		LHS&=&M@X &\textrm{(by def. of $\aabbortsymb$)}\\
		&\to_{R\beta}^{\leq 2}&M'@X&\textrm{(by Lemma \ref{lem:additional-compatibility-rules-1})}\\
		&=&RHS&\textrm{(by def. of $\aabbortsymb$)}
	\end{array}
	$$
	
	Case $C=C_1\supset C_2$.
	$$
	\begin{array}{rcll}
		LHS&=&\lb z.\aabbort M{C_2}&\textrm{(by def. of $\aabbortsymb$)}\\
		&\to_{R\beta}^*&\lb z.\aabbort {M'}{C_2}&\textrm{(by IH)}\\
		&=&RHS&\textrm{(by def. of $\aabbortsymb$)}
	\end{array}
	$$
	
	Case $C=C_1\wedge C_2$.
	$$
	\begin{array}{rcll}
		LHS&=&\ipair i{\aabbort M{C_i}}&\textrm{(by def. of $\aabbortsymb$)}\\
		&\to_{R\beta}^*&\ipair i{\aabbort {M'}{C_i}}&\textrm{(by IH twice)}\\
		&=&RHS&\textrm{(by def. of $\aabbortsymb$)}
	\end{array}
	$$
	
	Case $C=\forall X.C_0$.
	$$
	\begin{array}{rcll}
		LHS&=&\Lb X.\aabbort M{C_0}&\textrm{(by def. of $\aabbortsymb$)}\\
		&\to_{R\beta}^*&\Lb z.\aabbort {M'}{C_0}&\textrm{(by IH)}\\
		&=&RHS&\textrm{(by def. of $\aabbortsymb$)}
	\end{array}
	$$
	
	The more refined result when $R$ is a $\beta$-reduction comes from Lemma \ref{lem:additional-compatibility-rules-1}.1.(b).

\end{proof}



\begin{lem}\label{lem:term-subst-@}In $\fat$:
	\begin{enumerate}
		\item If $P=x$ and $N$ is an abstraction then\\ $[N/x](P@Q)\,\to_{\betai}\,([N/x]P)@([N/x]Q)$. Otherwise, $[N/x](P@Q)\,=\,([N/x]P)@([N/x]Q)$.\\ In particular $[N/x](P@Q)\,\to_{\betai}^=\,([N/x]P)@([N/x]Q)$.
		\item If $P=x$ and $N$ is a pair then $[N/x](\pproj iP)\,\to_{\betac}\,\pproj i{[N/x]P}$. Otherwise, $[N/x](\pproj iP)\,=\,\pproj i{[N/x]P}$. In particular $[N/x](\pproj iP)\,\to_{\betac}^=\,\pproj i{[N/x]P}$. 
		\item If $P=x$ and $N$ is a generalization then $[N/x](P@Z)\,\to_{\betaall}\,([N/x]P)@Z$. Otherwise, $[N/x](P@Z)\,=\,([N/x]P)@Z$. In particular $[N/x](P@Z)\to_{\betaall}^=([N/x]P)@Z$. 
	\end{enumerate}
\end{lem}
\begin{proof}
	Each item is proved by case analysis of $P$.
	
	1. If $P=\lb z.P'$, then $[N/x](P@Q)=[N/x][Q/z]P'\stackrel{(*)}{=}[[N/x]Q/z][N/x]P'\\=(\lb z.[N/x]P')@([N/x]Q)=([N/x]P)@([N/x]Q)$, where equality marked with $(*)$ is justified by the familiar substitution lemma in $\fat$. If $P$ is not an abstraction, then $([N/x]P)$ is an abstraction only if $P=x$ and $N=\lb y.N'$, say, in which case $[N/x](P@Q)=(\lb y.N')([N/x]Q)\to_{\betai}[[N/x]Q/y]N'=(\lb y.N')@([N/x]Q)=([N/x]P)@([N/x]Q)$. If $([N/x]P)$ is not an abstraction either, then $([N/x]P)@([N/x]Q)=[N/x](PQ)=[N/x](P@Q)$.
	
	2. and 3. The proof is done in a similar way considering pair and generalization respectively instead of abstraction.
\end{proof}

\begin{lem}\label{lem:type-subst-@}In $\fat$:
	\begin{enumerate}
		\item $[Y/X](P@Q)=([Y/X]P)@([Y/X]Q)$.
		\item $[Y/X](\pproj iP)=\pproj i{[Y/X]P}$.
		\item $[Y/X](P@Z)=([Y/X]P)@[Y/X]Z$.
	\end{enumerate}
\end{lem}
\begin{proof}
	Each item is proved by case analysis of $P$. The proof if similar to that of Lemma \ref{lem:term-subst-@}, but simpler, because whenever $P$ is not an abstraction (resp. pair, generalization), neither is $[Y/X]P$. 
\end{proof}
\begin{lem}\label{lem:term-subst-expansions} In $\fat$:
	\begin{enumerate}
		\item $[N/x](\llb y.M)=\llb y. [N/x]M$.
		\item $[N/x]\ppair{M_1}{M_2}=\ppair{[N/x]M_1}{[N/x]M_2}.$
	\end{enumerate}
\end{lem}
\begin{proof}
	The proof is straightforward from the usual substitution in $\fat$ and the definitions of $\llb$ and $\ppair{\cdot}{\cdot}$ respectively. 
\end{proof}

\begin{lem}\label{lem:subst-admissible-optimized-constructions} In $\fat$:
	\begin{enumerate}
		\item
		\begin{enumerate}
			\item {\small{$[N/z]\accase MxPyQC\to_{\beta}^*\accase{[N/z]M}x{[N/z]P}y{[N/z]Q}C$}}.
			\item {\small{$[Y/X]\accase MxPyQC=\\ =\accase{[Y/X]M}x{[Y/X]P}y{[Y/X]Q}{[Y/X]C}$}}.
		\end{enumerate}
		\item
		\begin{enumerate}
			\item $[N/z]\aabbort MC\to_{\beta}^*\aabbort{[N/z]M}C$.
			\item $[Y/X]\aabbort MC=\aabbort{[Y/X]M}{[Y/X]C}$.
		\end{enumerate}
		\item  Item 1.(a) (resp. 2.(a)) has a variant form that holds in a strong form (with syntactic identity), when $z\notin M,P,Q$ and $N$ is not an abstraction (resp. $z\notin M$ and $N$ arbitrary):
		\begin{enumerate}
			\item $[N/z]\accase Mx{P@z}y{Q@z}C=\accase{M}x{P@N}y{Q@N}C$.
			\item $[N/z]\aabbort MC=\aabbort{M}C$.
		\end{enumerate}
	\end{enumerate}
\end{lem}
\begin{proof}
	1(a). By induction on $C$.
	
	Case $C=X$.
	$$
	\begin{array}{rcll}
		LHS&=&[N/z]((M@X)@\pair{\lb x.P}{\lb y.Q})&\textrm{(by def. of $\accasesymb$)}\\
		&\to_{\beta}^=&([N/z](M@X))@(\pair{\lb x.[N/z]P}{\lb y.[N/z]Q})&\textrm{(by Lem. \ref{lem:term-subst-@})}\\
		&\to_{\beta}^{\leq 2}&(([N/z]M)@X))@\pair{\lb x.[N/z]P}{\lb y.[N/z]Q}&\textrm{(by Lem. \ref{lem:term-subst-@} and \ref{lem:additional-compatibility-rules-1})}\\
		&=&RHS&\textrm{(by def. of $\accasesymb$)}
	\end{array}
	$$
	
	Case $C=C_1\supset C_2$.
	$$
	\begin{array}{rcll}
		LHS&=&\lb w.[N/z](\accase Mx{P@w}y{Q@w}{C_2})&
		\\
		&\to_{\beta}^*&\lb w.\accase {[N/z]M}x{[N/z](P@w)}y{[N/z](Q@w)}{C_2}&\textrm{(by IH)}\\
		&\to_{\beta}^*&\lb w.\accase {[N/z]M}x{([N/z]P)@w}y{([N/z]Q)@w}{C_2}&
		\\
		&=&RHS&
	\end{array}
	$$
	The first and fourth steps above follow from the definition of $\accasesymb$, and the third step relies on Lemmas \ref{lem:term-subst-@} and \ref{lem:additional-compatibility-rules-2}.\\
	Case $C=C_1\wedge C_2$.
	$$
	\begin{array}{rcll}
		LHS&=&\ipair{i}{[N/z]\accase Mx{\pproj iP}y{\pproj iQ}{C_i}}&
		\\
		&\to_{\beta}^*&\ipair{i}{\accase {[N/z]M}x{[N/z](\pproj iP)}y{[N/z](\pproj iQ)}{C_i}}&
		\\
		&\to_{\beta}^*&\ipair{i}{\accase {[N/z]M}x{\pproj i{([N/z]P)}}y{\pproj i{([N/z]Q)}}{C_i}}&
		\\
		&=&RHS&
	\end{array}
	$$
	The first and fourth steps above follow from the definition of $\accasesymb$, the second by applying the IH twice and the third relies on Lemmas \ref{lem:term-subst-@} and \ref{lem:additional-compatibility-rules-2}.\\
	Case $C=\forall C_0$.
	$$
	\begin{array}{rcll}
		LHS&=&\Lb X.[N/z]\accase Mx{P@X}y{Q@X}{C_0}&
		\\
		&\to_{\beta}^*&\Lb X.\accase {[N/z]M}x{[N/z](P@X)}y{[N/z](Q@X)}{C_0}&\textrm{(by IH)}\\
		&\to_{\beta}^*&\Lb X.\accase {[N/z]M}x{([N/z]P)@X}y{([N/z]Q)@X}{C_0}&
		\\
		&=&RHS&
	\end{array}
	$$
	The first and fourth steps above follow from the definition of $\accasesymb$, and the third step relies on Lemmas \ref{lem:term-subst-@} and \ref{lem:additional-compatibility-rules-2}.
	
	1 (b). Straightforward induction on $C$, using Lemma \ref{lem:type-subst-@}.
	
	2 (a). By induction on $C$.
	
	Case $C=X$.
	$$
	\begin{array}{rcll}
		LHS&=&[N/z](M@X) &\textrm{(by def. of $\aabbortsymb$)}\\
		&\to_{\beta}^=&([N/z]M)@X&\textrm{(by Lemma \ref{lem:term-subst-@})}\\
		&=&RHS&\textrm{(by def. of $\aabbortsymb$)}
	\end{array}
	$$
	
	Case $C=C_1\supset C_2$.
	$$
	\begin{array}{rcll}
		LHS&=&\lb z.[N/z]\aabbort M{C_2}&\textrm{(by def. of $\aabbortsymb$)}\\
		&\to_{\beta}^*&\lb z.\aabbort {[N/z]M}{C_2}&\textrm{(by IH)}\\
		&=&RHS&\textrm{(by def. of $\aabbortsymb$)}
	\end{array}
	$$
	
	Case $C=C_1\wedge C_2$.
	$$
	\begin{array}{rcll}
		LHS&=&\ipair i{[N/z]\aabbort M{C_i}}&\textrm{(by def. of $\aabbortsymb$)}\\
		&\to_{\beta}^*&\ipair i{\aabbort {[N/z]M}{C_i}}&\textrm{(by IH twice)}\\
		&=&RHS&\textrm{(by def. of $\aabbortsymb$)}
	\end{array}
	$$
	
	Case $C=\forall X.C_0$.
	$$
	\begin{array}{rcll}
		LHS&=&\Lb X.[N/z]\aabbort M{C_0}&\textrm{(by def. of $\aabbortsymb$)}\\
		&\to_{\beta}^*&\Lb z.\aabbort {[N/z]M}{C_0}&\textrm{(by IH)}\\
		&=&RHS&\textrm{(by def. of $\aabbortsymb$)}
	\end{array}
	$$
	
	2 (b). Straightforward induction on $C$, using Lemma \ref{lem:type-subst-@}.
	
	3 (a). Go through the proof of 1 (a) and check that, when $N$ is not an abstraction, all applications of Lemma \ref{lem:term-subst-@} produce syntactic equality - therefore so do all applications of IH. Hence we already have 
	
	$LHS=\accase{[N/z]M}x{[N/z](P@z)}y{[N/z](Q@z)}C$. Applying again Lemma \ref{lem:term-subst-@} twice, and using $z\notin M,P,Q$, we obtain the desired RHS.\footnote{\label{fn:counter-example}Of course it would be more elegant to state item 3(a) with arbitrary $N$; it turns out that such statement fails for arbitrary $N$ (counter-example: $N$ an abstraction, $P$ or $Q$ is $\lb w.w$, and $C=C_1\supset C_2$).}
	
	3 (b). Go through the proof of 2 (a) and check that, in this particular case, due to $z\notin M$, the single application of Lemma \ref{lem:term-subst-@} produce syntactic equality - therefore so do all applications of IH.
\end{proof}

\subsection{Admissible reduction rules}

We now see that all reduction rules of $\ipc$ are admissible in $\fat$ in the form of reduction (possibly syntactic equality). 
\begin{lem}[Admissible $\betas, \etas$ with $\connective=\supset,\wedge$]\label{lem:admissible-beta-eta-implication-conj} In $\fat$:
	\begin{enumerate}
		\item $(\llb x.M)@N\to_{\beta}[N/x]M$.
		\item $\ppair{M_1}{M_2}@i\to_{\beta}M_i$.
		\item $\llb x.M@x\to^+_{\eta}M,$ if $x\notin M$.
		\item $\ppair{M@1}{M@2}\to^+_{\eta}M$.
	\end{enumerate}
\end{lem}
\begin{proof}
	Proof of 1. $(\llb x.M)@N=(\lb x. M)N \to_{\betai}[N/x]M$, where the first equality if by item 3 of Lemma \ref{lem:admissible-eta-reduction}.
	
	Proof of 2. $\ppair{M_1}{M_2}@i=\pair{M_1}{M_2}i\to_{\betac}M_i$, where the first equality if by item 4 of Lemma \ref{lem:admissible-eta-reduction}.
	
	Proof of 3. Suppose $x\notin M$. $\llb x.M@x\to_{\etai}\lb x. M@x\to^=_{\etai}M$, where use is made of items 1 and 5 of Lemma \ref{lem:admissible-eta-reduction}.
	
	Proof of 4.
	$\ppair{M@1}{M@2}\to_{\etac}\pair{M@1}{M@2}\to^=_{\etac}M$, where use is made of items 2 and 6 of Lemma \ref{lem:admissible-eta-reduction}.
\end{proof}

%
%
%
%
\begin{lem}[Admissible $\betad$]\label{lem:optimized-admissible-beta-disjunction} In $\fat$: $\accase{\aiinjn iNAB}{x_1}{P_1}{x_2}{P_2}C\to^+_{\beta\eta}[N/x_i]P_i$.
\end{lem}
\begin{proof} By induction on $C$. 
	
	Case $C=Y$.
	$$
	\begin{array}{cll}
		&LHS&\\
		=&(\Lb X.\lb z^{(A\supset X)\wedge(B\supset X)}.\proj iz N)@Y@\pair{\lb x_1.P_1}{\lb x_2.P_2}&\textrm{(by def. of $\accasesymb$)}\\
		=&(\lb z^{(A\supset Y)\wedge(B\supset Y)}.\proj iz N)\pair{\lb x_1.P_1}{\lb x_2.P_2}&
		\\
		=&\proj i{\pair{\lb x_1.P_1}{\lb x_2.P_2}}N&
		\\
		\to_{\betac}&(\lb x_i.P_i)N&\\
		\to_{\betai}&[N/x_i]P_i
	\end{array}
	$$
	
	
	The second step above follows from the definition of $@$ and the fact that $X\notin A,B,N$; the third step follows also from the definition of $@$ and the fact that $z\notin N$.
	
	Case $C=D_1\supset D_2$.
	$$
	\begin{array}{cll}
		&LHS&\\
		=&\lb z.\accase{\aiinjn iNAB}{x_1}{P_1 @z}{x_2}{P_2 @z}{D_2}&\textrm{(by def. of $\accasesymb$)}\\
		\to^+_{\beta\eta}&\lb z.[N/x_i](P_i @z)&\textrm{(by IH)}\\
		\to_{\betai}^=&\lb z.([N/x_i]P_i)@([N/x_i]z)&\textrm{(by Lemma \ref{lem:term-subst-@})}\\
		=&\lb z.([N/x_i]P_i)@z&\textrm{(since $z\neq x_i$)}\\
		\to_{\etai}^=&[N/x_i]P_i&\textrm{(by Lemma \ref{lem:admissible-eta-reduction})}
	\end{array}
	$$
	
	Case $C=D_1\wedge D_2$. Similar, uses $\betac$ and $\etac$.
	
	Case $C=\forall Y.D$.
	$$
	\begin{array}{cll}
		&LHS&\\
		=&\Lb Y.\accase{\aiinjn iNAB}{x_1}{P_1 @Y}{x_2}{P_2 @Y}{D}&\textrm{(by def. of $\accasesymb$)}\\
		\to^+_{\beta\eta}&\Lb Y.[N/x_i](P_i @Y)&\textrm{(by IH)}\\
		\to_{\betaall}^=&\Lb Y.([N/x_i]P_i)@Y&\textrm{(by Lemma \ref{lem:term-subst-@})}\\
		\to_{\etaall}^=&[N/x_i]P_i&\textrm{(by Lemma \ref{lem:admissible-eta-reduction})}
	\end{array}
	$$
\end{proof}

\begin{lem}[Admissible $\etad$]\label{lem:optimized-admissible-eta-disjunction} In $\fat$:
	$$\accase M{x^A}{\aiinjn 1xAB}{y^B}{\aiinjn 2yAB}{A\dvee B}\to^*_{\eta} M.$$
\end{lem}
\begin{proof} Let $C:=((A\supset X)\wedge(B\supset X))\supset X$.
	
	$$
	\begin{array}{cll}
		&LHS&\\
		=&\Lb X.\accase{M}{x}{(\Lb Y.\lb z.\proj 1 z x)@X}{y}{(\Lb Y.\lb z.\proj 2 z y)@X}{C}&
		\\
		=&\Lb X.\lb w.\accase{M}{x}{(\Lb Y.\lb z.\proj 1 z x)@X@w}{y}{(\Lb Y.\lb z.\proj 2 z y)@X@w}{X}&
		\\
		=&\Lb X.\lb w.M@X@\pair{\lb x.(\Lb Y.\lb z.\proj 1 z x)@X@w}{\lb y.(\Lb Y.\lb z.\proj 2 z y)@X@w}&
		\\
		=&\Lb X.\lb w.{M}@X@\pair{\lb{x}.{\proj 1 w x}}{\lb{y}.{\proj 2 w y}}&
		\\
		\to_{\etai}^*&\Lb X.\lb w.{M}@X@\pair{{\proj 1 w}}{{\proj 2 w }}&
		\\
		\to_{\etac}^*&\Lb X.\lb w.{M}@X@w&
		\\
		\to_{\etai}^=&\Lb X.{M}@X&
		\\
		\to_{\etaall}^=&M&
	\end{array}
	$$
	
	The first four steps (equalities) above follow from the definitions, including those for $\accasesymb$ and $@$. Steps five and six are justified by Lemma \ref{lem:additional-compatibility-rules-1} and the last two steps by Lemma \ref{lem:admissible-eta-reduction}.
	
\end{proof}

\begin{lem}[Admissible $\pis$, for $\connective=\supset,\wedge,\forall$]\label{lem:optimized-admissible-commutative-conversions} In $\fat$:
	\begin{enumerate}
		\item $(\accase MxPyQ{C\supset D})@N\to_{\beta}^*\accase Mx{P@N}y{Q@N}D$. The result holds with syntactic equality when $N$ is not an abstraction.
		\item $\pproj i{\accase MxPyQ{C_1\wedge C_2}}=\accase Mx{\pproj iP}y{\pproj iQ}{C_i}$.
		\item $(\accase MxPyQ{\forall X.C})@Y=\accase Mx{P@Y}y{Q@Y}{[Y/X]C}$.
	\end{enumerate}
\end{lem}
\begin{proof} Proof of 1.
	$$
	\begin{array}{cll}
		&LHS&\\
		=&(\lb z.\accase Mx{P@z}y{Q@z}D)@N&\textrm{(by def. of $\accasesymb$)}\\
		=&[N/z]\accase Mx{P@z}y{Q@z}D&\textrm{(by def. of $@$)}\\
		\to_{\beta}^*&\accase {[N/z]M}x{[N/z](P@z)}y{[N/z](Q@z)}D&\textrm{(by Lemma \ref{lem:subst-admissible-optimized-constructions})}\\
		=&\accase {[N/z]M}x{([N/z]P)@N}y{([N/z]Q)@N}D&\textrm{(by Lemma \ref{lem:term-subst-@})}
		\\
		=&RHS&\textrm{(since $z\notin M,P,Q$)}
	\end{array}
	$$
	In this calculation we used item 1(a) of Lemma \ref{lem:subst-admissible-optimized-constructions} and when applying Lemma \ref{lem:term-subst-@} we used the fact that $P\neq z$ and $Q\neq z$. When $N$ is not an abstraction, item 3 (a) of the same lemma gives $[N/z]\accase Mx{P@z}y{Q@z}D=RHS$.
	
	Proof of 2.
	$$
	\begin{array}{rcll}
		LHS&=&\pproj i{\left(\langle\accase Mx{\pproj jP}y{\pproj jQ}{C_j}\rangle_{j=1,2}\right)}&\textrm{(by def. of $\accasesymb$)}\\
		&=&RHS&\textrm{(by def. of $@$)}
	\end{array}
	$$
	Proof of 3.
	$$
	\begin{array}{cll}
		&LHS&\\
		=&(\Lb X.\accase Mx{P@X}y{Q@X}C)@Y&
		\\
		=&[Y/X]\accase Mx{P@X}y{Q@X}C&
		\\
		=&\accase {[Y/X]M}x{[Y/X](P@X)}y{[Y/X](Q@X)}{[Y/X]C}&
		\\
		=&\accase {[Y/X]M}x{([Y/X]P)@Y}y{([Y/X]Q)@Y}{[Y/X]C}&
		\\
		=&RHS&
	\end{array}
	$$
	In the first two equalities above, we used the definitions of $\accasesymb$ and $@$, respectively. The third equality follows from Lemma \ref{lem:subst-admissible-optimized-constructions}; the fourth equality from Lemma \ref{lem:type-subst-@}, and for the last equality, note that $X\notin M,P,Q$.   \end{proof}

\begin{lem}[Admissible $\abs$, for $\connective=\supset,\wedge,\forall$]\label{lem:optimized-admissible-absurdity-conversions} In $\fat$:
	\begin{enumerate}
		\item $(\aabbort M{A\supset B})@N=\aabbort MB$.
		\item $\pproj i{\aabbort M{A_1\wedge A_2}}=\aabbort M{A_i}$, $i=1,2$.
		\item $(\aabbort M{\forall X.A})@Y=\aabbort M{[Y/X]A}$.
	\end{enumerate}
\end{lem}
\begin{proof}
	Proof of 1.
	$$
	\begin{array}{rcll}
		LHS&=&(\lb z^A.\aabbort MB)@N&\textrm{(by def. of $\aabbortsymb$)}\\
		&=&[N/z]\aabbort MB&\textrm{(by def of $@$)}\\
		&=&RHS&\textrm{(by item 3(b) of Lem. \ref{lem:subst-admissible-optimized-constructions}, since $z\notin M$)}
	\end{array}
	$$
	
	Proof of 2.
	$$
	\begin{array}{rcll}
		LHS&=&\pproj i{\pair{\aabbort M{A_1}}{\aabbort M{A_2}}}&\textrm{(by def. of $\aabbortsymb$)}\\
		&=&RHS&\textrm{(by def of $@$)}
	\end{array}
	$$
	
	Proof of 3.
	$$
	\begin{array}{rcll}
		LHS&=&\Lambda X.(\aabbort MA)@Y&\textrm{(by def. of $\aabbortsymb$)}\\
		&=&[Y/X]\aabbort MA&\textrm{(by def of $@$)}\\
		&=&\aabbort{[Y/X]M}{[Y/X]A}&\textrm{(by item 2 (b) of Lemma \ref{lem:subst-admissible-optimized-constructions})}\\
		&=&RHS&\textrm{(since $X\notin M$)}
	\end{array}
	$$
\end{proof}

\begin{lem}[Admissible $\abd$]\label{lem:optimized-admissible-absurdity-disjunction} In $\fat$:
	$$
	\accase{\aabbort M{A\dvee B}}xPyQC=\aabbort MC.
	$$
\end{lem}
\begin{proof}
	By induction on $C$.
	
	Case $C=X$.
	$$
	\begin{array}{cll}
		&LHS&\\
		=&\aabbort M{A\vee B}@X@\pair{\lb x.P}{\lb y.Q}&\textrm{(by def. of $\accasesymb$)}\\
		=&(\Lambda Y\lambda z^{(A\supset Y)\wedge(B\supset Y)}.M@Y)@X@\pair{\lb x.P}{\lb y.Q}&\textrm{(by def. of $\aabbortsymb$)}\\
		=&([X/Y](\lambda z^{(A\supset Y)\wedge(B\supset Y)}.M@Y))@\pair{\lb x.P}{\lb y.Q}&\textrm{(by def. of $@$)}\\
		=&(\lambda z^{(A\supset X)\wedge(B\supset X)}.M@X)@\pair{\lb x.P}{\lb y.Q}&\textrm{(by Lemma \ref{lem:type-subst-@})}\\
		=&M@X&\textrm{(by def of $@$ and $z\notin M$)}\\
		=&RHS&\textrm{(by def. of $\aabbortsymb$)}
	\end{array}
	$$
	
	When applying Lemma \ref{lem:type-subst-@} above, we are also using the fact that $Y\notin M,A,B$.
	
	Case $C=C_1\supset C_2$.
	$$
	\begin{array}{rcll}
		LHS&=&\lb z^{C_1}.\accase{\aabbort M{A\vee B}}x{P@z}y{Q@z}{C_2}&
		\\
		&=&\lb z^{C_1}.\aabbort M{C_2}&\textrm{(by IH)}\\
		&=&RHS&
	\end{array}
	$$
	The first and last equalities above come from the definitions of $\accasesymb$ and $\aabbortsymb$, respectively.
	
	Case $C=C_1\wedge C_2$.
	$$
	\begin{array}{rcll}
		LHS&=&\langle\accase{\aabbort M{A\vee B}}x{\pproj iP}y{\pproj iQ}{C_i}\rangle_{i=1,2}&\textrm{(def. of $\accasesymb$)}\\
		&=&\pair{\aabbort M{C_1}}{\aabbort M{C_2}}&\textrm{(by IH twice)}\\
		&=&RHS&\textrm{(def. of $\aabbortsymb$)}\\
	\end{array}
	$$
	
	Case $C=\forall Y.D$.
	$$
	\begin{array}{rcll}
		LHS&=&\Lambda Y.\accase{\aabbort M{A\vee B}}x{P@Y}y{Q@Y}D&\textrm{(def. of $\accasesymb$)}\\
		&=&\Lambda Y.\aabbort MD&\textrm{(by IH)}\\
		&=&RHS&\textrm{(def. of $\aabbortsymb$)}\\
	\end{array}
	$$

\end{proof}

\begin{lem}[Admissible $\aba$]\label{lem:optimized-admissible-absurdity-absurdity} In $\fat$:
	$$
	\aabbort{\aabbort M{\perp}}A=\aabbort MA.
	$$
\end{lem}
\begin{proof}
	By induction on $A$.
	
	Case $A=Y$.
	$$
	\begin{array}{rcll}
		LHS&=&(\Lambda X.M@X)@Y&\textrm{(by def. of $\aabbortsymb$)}\\
		&=&[Y/X](M@X)&\textrm{(by def. of $@$)}\\
		&=&M@Y&\textrm{(by Lemma \ref{lem:type-subst-@} and $X\notin M$)}\\
		&=&RHS&\textrm{(by def. of $\aabbortsymb$)}\\
	\end{array}
	$$
	
	Case $A=B\supset C$.
	$$
	\begin{array}{rcll}
		LHS&=&\lb z^B.\aabbort{\aabbort M\perp}C&\textrm{(by def. of $\aabbortsymb$)}\\
		&=&\lb z^B.\aabbort MC&\textrm{(by IH)}\\
		&=&RHS&\textrm{(by def. of $\aabbortsymb$)}\\
	\end{array}
	$$
	Cases $A=B_1\wedge B_2$ and $A=\forall Y.B$ follow similarly by IH and definition of $\aabbortsymb$.
\end{proof}

\begin{lem}[Admissible $\pid$-reduction]\label{lem:admissible-pi-disjunction-reduction} In $\fat$:
	$$
	\begin{array}{l}\accase{\accase M{x_1}{P_1}{x_2}{P_2}{B_1\dvee B_2}}{y_1}{Q_1}{y_2}{Q_2}{C}=\\
		\qquad\accase M{x_1}{\accase{P_1}{y_1}{Q_1}{y_2}{Q_2}C}{x_2}{\accase{P_2}{y_1}{Q_1}{y_2}{Q_2}C}C.
	\end{array}$$
\end{lem}
\begin{proof} By induction on $C$. We follow closely the proof of lemma 11 in paper \cite{EspiritoSantoFerreira2020}. The calls to lemma 7 of that paper, which generated there the $\beta$-reduction steps in the ``wrong'' direction, are replaced here by equalities justified by Lemma \ref{lem:optimized-admissible-commutative-conversions}.
	
	Case $C=Y$.The LHS term is, by definition of $\accasesymb$,
	$$
	(\Lb X.\lb w.(M@X)@\pair{\lb x_1.P_1 @X@w}{\lb x_2.P_2 @X@w})@Y@\pair{\lb y_1.Q_1}{\lb y_2.Q_2},
	$$
	where the variable $w$ has type $(B_1\supset X)\wedge(B_2\supset X)$. By definition of $@$, it is equal to
	$$
	(\lb w.[Y/X]((M@X)@\pair{\lb x_1.P_1 @X@w}{\lb x_2.P_2 @X@w}))@\pair{\lb y_1.Q_1}{\lb y_2.Q_2}.
	$$
	Due to Lemma \ref{lem:type-subst-@} and $X\notin M,P_1,P_2$, this term is equal to
	$$
	(\lb w^{(B_1\supset Y)\wedge(B_2\supset Y)}.(M@Y)@\pair{\lb x_1.P_1 @Y@w}{\lb x_2.P_2 @Y@w})\pair{\lb y_1.Q_1}{\lb y_2.Q_2},
	$$
	which, in turn, again by definition of $@$, is equal to
	$$
	[P/w]((M@Y)@\pair{\lb x_1.P_1 @Y@w}{\lb x_2.P_2 @Y@w}).
	$$
	where $P=\pair{\lb y_1.Q_1}{\lb y_2.Q_2}$. Due to Lemma \ref{lem:term-subst-@} and $M@Y\neq w$ (recall $w\notin M$), this term is equal to
	$$
	([P/w](M@Y)@\pair{\lb x_1.[P/w](P_1 @Y@w)}{\lb x_2.[P/w](P_2 @Y@w)}).
	$$
	Applying Lemma \ref{lem:term-subst-@} five times, having in mind that none of the terms $M$, $P_1$, $P_1@Y$, $P_2$, $P_2@Y$ is $w$ (recall $w\notin M,P_1,P_2$), we conclude that the term is equal to
	$$
	(M@Y)@\pair{\lb x_1.(P_1 @Y@\pair{\lb y_1.Q_1}{\lb y_2.Q_2})}{\lb x_2.(P_2 @Y@\pair{\lb y_1.Q_1}{\lb y_2.Q_2})}.
	$$
	Now, by definition of $\accasesymb$, this term is
	$$
	(M@Y)@\pair{\lb x_1.\accase{P_1}{y_1}{Q_1}{y_2}{Q_2}{Y}}{\lb x_2.\accase{P_2}{y_1}{Q_1}{y_2}{Q_2}{Y}},
	$$
	which is the RHS term, again by definition of $\accasesymb$.
	
	Case $C=C_1\supset C_2$. The LHS term is, by definition of $\accasesymb$,
	$$\lb z^{C_1}.\accase{\accase M{x_1}{P_1}{x_2}{P_2}{B_1\vee B_2}}{y_1}{Q_1@z}{y_2}{Q_2@z}{C_2}, $$
	which, by IH, is equal to
	{\footnotesize{$$
	\lb z^{C_1}.\accase M{x_1}{\accase{P_1}{y_1}{Q_1@z}{y_2}{Q_2@z}{C_2}}{x_2}{\accase{P_2}{y_1}{Q_1@z}{y_2}{Q_2@z}{C_2}}{C_2}.
	$$}}
	By item 1 of Lemma \ref{lem:optimized-admissible-commutative-conversions} and since $z$ is not an abstraction, this term is
		{\footnotesize{$$
	\lb z^{C_1}.\accase M{x_1}{(\accase{P_1}{y_1}{Q_1}{y_2}{Q_2}{C})@z}{x_2}{(\accase{P_2}{y_1}{Q_1}{y_2}{Q_2}{C})@z}{C_2},
	$$}}
	which is the RHS term, by definition of $\accasesymb$.
	
	Case $C=C_1\wedge C_2$. The LHS terms is, by definition of $\accasesymb$,
	$$
	\ipair i{\accase{\accase M{x_1}{P_1}{x_2}{P_2}{B_1\vee B_2}}{y_1}{\pproj i{Q_1}}{y_2}{\pproj i{Q_2}}{C_i}},
	$$
	which, by application of IH twice, is equal to
		{\footnotesize{$$
	\ipair i{\accase M{x_1}{\accase{P_1}{y_1}{\pproj i{Q_1}}{y_2}{\pproj i{Q_2}}{C_i}}{x_2}{\accase{P_2}{y_1}{\pproj i{Q_1}}{y_2}{\pproj i{Q_2}}{C_i}}{C_i}}.
	$$}}
	By item 2 of Lemma \ref{lem:optimized-admissible-commutative-conversions}, this term is equal to
		{\footnotesize{$$
	\ipair i{\accase M{x_1}{\pproj i{\accase{P_1}{y_1}{Q_1}{y_2}{Q_2}{C_i}}}{x_2}{\pproj i{\accase{P_2}{y_1}{Q_1}{y_2}{Q_2}{C_i}}}{C_i}},
	$$}}
	which is the RHS term, by definition of $\accasesymb$.
	
	Case $C=\forall Y.D$. The LHS term is, by definition of $\accasesymb$,
	$$
	\Lambda Y.\accase{\accase M{x_1}{P_1}{x_2}{P_2}{B_1\vee B_2}}{y_1}{Q_1@Y}{y_2}{Q_2@Y}{D}
	$$
	which, by application of IH, is equal to
		{\footnotesize{$$
	\Lambda Y.\accase M{x_1}{\accase{P_1}{y_1}{Q_1@Y}{y_2}{Q_2@Y}D}{x_2}{\accase{P_2}{y_1}{Q_1@Y}{y_2}{Q_2@Y}D}{D}.
	$$}}
	By item 3 of Lemma \ref{lem:optimized-admissible-commutative-conversions}, this term is equal to
		{\small{$$
	\Lambda Y.\accase M{x_1}{\accase{P_1}{y_1}{Q_1}{y_2}{Q_2}{C} @Y}{x_2}{\accase{P_2}{y_1}{Q_1}{y_2}{Q_2}{C} @Y}{D},
	$$}}
	which is the RHS term, by definition of $\accasesymb$.
\end{proof}

\begin{lem}[Admissible $\pia$-reduction]\label{lem:admissible-pi-absurdity-reduction} In $\fat$:
		{\small{$$
	\aabbort{\accase MxPyQ{\dperp}}C=\accase Mx{\aabbort PC}y{\aabbort QC}C.
	$$}}
\end{lem}
\begin{proof} By induction on $C$. We follow closely the proof of lemma 12 in paper \cite{EspiritoSantoFerreira2020}. The calls to lemma 8 of that paper, which generated there the $\beta$-reduction steps in the ``wrong'' direction, are here replaced by equalities justified by Lemma \ref{lem:optimized-admissible-absurdity-conversions}.
	
	Case $C=Y$.
	$$
	\begin{array}{cll}
		&LHS&\\
		=&\accase MxPyQ{\dperp}@Y&\textrm{(by def. of $\aabbortsymb$)}\\
		=&(\Lambda X.M@X@\pair{\lb x.P@X}{\lb y.Q@X})@Y&\textrm{(by def. of $\accasesymb$)}\\
		=&[Y/X](M@X@\pair{\lb x.P@X}{\lb y.Q@X})&\textrm{(by def. of $@$)}\\
		=&M@Y@\pair{\lb x.P@Y}{\lb y.Q@Y}&\textrm{(Lem. \ref{lem:type-subst-@} and $X\notin M,P,Q$)}\\
		=&M@Y@\pair{\lb x.\aabbort PY}{\lb y.\aabbort QY}&\textrm{(by def. of $\aabbortsymb$)}\\
		=&RHS&\textrm{(by def. of $\accasesymb$)}
	\end{array}
	$$
	
	Case $C=C_1\supset C_2$. The LHS term is, by definition of $\aabbortsymb$,
	$$
	\lb z^{C_1}.\aabbort{\accase MxPyQ{\perp}}{C_2},
	$$
	which, by IH, is equal to
	$$
	\lb z^{C_1}.\accase Mx{\aabbort P{C_2}}y{\aabbort Q{C_2}}{C_2}.
	$$
	Due to item 1 of Lemma \ref{lem:optimized-admissible-absurdity-conversions}, this term is equal to
	$$
	\lb z^{C_1}.\accase Mx{\aabbort P{C_1\supset C_2}@z}y{\aabbort Q{C_1\supset C_2}@z}{C_2},
	$$
	which is the RHS term, by definition of $\accasesymb$.
	
	Case $C=C_1\wedge C_2$. The LHS term, is, by definition of $\aabbortsymb$,
	$$
	\ipair i{\aabbort{\accase MxPyQ{\perp}}{C_i}},
	$$
	which, by IH applied twice, is equal to
	$$
	\ipair i{\accase Mx{\aabbort P{C_i}}y{\aabbort Q{C_i}}{C_i}}.
	$$
	Due to item 2 of Lemma \ref{lem:optimized-admissible-absurdity-conversions}, this term is equal to
	$$
	\ipair i{\accase Mx{\pproj i{\aabbort P{C_1\wedge C_2}}}y{\pproj i{\aabbort Q{C_1\wedge C_2}}}{C_i}},
	$$
	which is the RHS term, by definition of $\accasesymb$.
	
	Case $C=\forall Y.D$. The LHS term is, by definition of $\aabbortsymb$,
	$$
	\Lambda Y.\aabbort{\accase MxPyQ{\perp}}D,
	$$
	which, by IH, is equal to
	$$
	\Lambda Y.\accase Mx{\aabbort PD}y{\aabbort QD}{D}.
	$$
	Due to item 3 of Lemma \ref{lem:optimized-admissible-absurdity-conversions}, this term is equal to
	$$
	\Lambda Y.\accase Mx{\aabbort P{\forall Y.D}@Y}y{\aabbort Q{\forall Y.D}@Y}{D},
	$$
	which is the RHS term, by definition of $\accasesymb$.
\end{proof}



\subsection{Proof of the simulation theorem}

We return to the study of the optimized translation.

\begin{lem}\label{lem:optimized-translation-commutes-subst}
	$[\om N/x]\om M\to_{\beta}^*\om{([N/x]M)}$.
\end{lem}
\begin{proof}
	The proof is by induction on $M$. The base case ($M=y$) is immediate. Let us analyse the other cases.
	
	Case $M=\lb y^A. M_0$.
	$$
	\begin{array}{rcll}
		
		LHS&=&[\om N/x]\llb y^{\om A}. \om M_0&\textrm{(by def. of $\om{(\cdot)}$)}\\
		&=&\llb y^{\om A}.[\om N/x]\om M_0&\textrm{(by item 1. of Lemma \ref{lem:term-subst-expansions})}\\
		&\to^*_{\beta}&\llb y^{\om A}. \om {([N/x]M_0)}&\textrm{(by IH)}\\
		&=&RHS&\textrm{(by defs. of $\am{(\cdot)}$ and subst.)}
	\end{array}
	$$
	
	The case $M=\pair{M_0}{M_1}$ is entirely similar, using item 2. of Lemma \ref{lem:term-subst-expansions} and the induction hypothesis twice.  
	
	Case $M=M_0M_1$.
	$$
	\begin{array}{rcll}
		LHS&=&[\om N/x](\om M_0 @\om M_1) &\textrm{(by def. of $\om{(\cdot)}$)}\\
		&\to^=_{\beta}&([\om N/x]\om M_0)@([\om N/x]\om M_1) &\textrm{(by item 1. of Lemma \ref{lem:term-subst-@})}\\
		&\to^*_{\beta}&\om {([N/x]M_0)}@\om {([N/x]M_1)}&\textrm{(by IH)}\\
		&=&RHS&\textrm{(by defs. of $\am{(\cdot)}$ and subst.)}
	\end{array}
	$$
	
	The case $M=M_0i$ is entirely similar, using item 2.\ of Lemma \ref{lem:term-subst-@} and the induction hypothesis.  
	
	The case $M=\injn i{M_0}{A}{B}$ follows easily by the definitions of $\om {(\cdot)}$ and $\aiinjn i{\om {M_0}}{\om A}{\om B}$, by substitution in $\fat$ and by the induction hypothesis. 
	
	Case $M=\case {M_0}{y_1}{P_1}{y_2}{P_2}C$, with $P_1,P_2:C$.
	$$
	\begin{array}{cll}
		&LHS&\\
		=&[\om N/x]\accase{\om{M_0}}{y_1}{\om{P_1}}{y_2}{\om{P_2}}{\om C}&\textrm{(by def. of $\om{(\cdot)}$)}\\
		\to^*_{\beta}&\accase{[\om N/x]\om{M_0}}{y_1}{[\om N/x]\om{P_1}}{y_2}{[\om N/x]\om{P_2}}{\om C}&
		\\
		\to^*_{\beta}&\accase{\om{([N/x]M_0)}}{y_1}{\om{([N/x]{P_1})}}{y_2}{\om{([N/x]P_2)}}{\om C}&\textrm{(by IH)}\\
		=&RHS&
	\end{array}
	$$
	
	The second step above follows from item 1.(a) of Lemma \ref{lem:subst-admissible-optimized-constructions}; and the final step is a result of the definitions of $\om{(\cdot)}$ and substitution.
	
	Case $M=\abort {M_0}A$.
	$$
	\begin{array}{rcll}
		LHS&=&[\om N/x]\aabbort{\om{M_0}}{\om A}&\textrm{(by def. of $\om{(\cdot)}$)}\\
		&	\to^*_{\beta}&\aabbort{[\om N/x]\om{M_0}}{\om A}&\textrm{(by item 2.(a) of Lemma \ref{lem:subst-admissible-optimized-constructions})}\\
		&	\to^*_{\beta}&\aabbort{\om{([N/x]M_0)}}{\om{A}}&\textrm{(by IH)}\\
		&=&RHS&\textrm{(by defs. of $\om{(\cdot)}$ and subst.)}
	\end{array}
	$$
\end{proof}

\begin{thm}[Simulation]\label{thm:simulation}Let $R$ be a reduction rule of $\ipc$ given in Fig.\ref{fig:reduction-rules}.
	\begin{enumerate}
		\item Case $R\in\{\betai,\betac\}$. If $M\to_R N$ in $\ipc$, then $\om{M}\to_{\beta}^*\om{N}$ in $\fat$.
		\item Case $R\in\{\betad,\etad, \etai,\etac \}$. If $M\to_R N$ in $\ipc$, then $\om{M}\to_{\beta\eta}^*\om{N}$ in $\fat$.
		\item Case $R\in\{\pii,\pic,\pid,\pia,\abi,\abc,\abd,\aba\}$. If $M\to_R N$ in $\ipc$, then $\om{M}=\om{N}$ in $\fat$.
	\end{enumerate}
\end{thm}
\begin{proof} In fact, in the last item, concerning commuting conversion rules, for the case $R=\pii$, our analysis so far only allow us to prove $\om{M}\to_{\beta}^*\om{N}$, and this is what we prove here. The stronger result, that indeed no $\beta$-reduction step exists between $\om M$ and $\om N$ in that case as well, requires a deeper analysis, to be carried out afterwards, in the next section.
	
	For each $R$, we do the proof by induction on $M\to_R N$. Let us check the base cases. 
	
	Case $\betai$:
	$$
	\begin{array}{rcll}
		\om{((\lb x.M)N)}&=&(\llb x. \om M)@\om N &\textrm{(by def. of $\om{(\cdot)}$)}\\
		&\to_{\betai}& [\om N/x]\om M& \textrm{(by Lemma \ref{lem:admissible-beta-eta-implication-conj})}\\
		&\to^*_{\beta}&  \om {([N/x]M)}&\textrm{(by Lemma \ref{lem:optimized-translation-commutes-subst})}
	\end{array}
	$$
	
	Case $\betac$:
	$$
	\begin{array}{rcll}
		\proj{i}{\pair{M_1}{M_2}}&=&\ppair{\om M_1}{\om M_2}@i &\textrm{(by def. of $\om{(\cdot)}$)}\\
		&\to_{\betac}& \om M_i& \textrm{(by Lemma \ref{lem:admissible-beta-eta-implication-conj})}
	\end{array}
	$$
	Case $\betad$:
	$$
	\begin{array}{cll}
		&\om{(\cse{\injn iM{A_1}{A_2}}{x_1^{A_1}}{P_1}{x_2^{A_2}}{P_2})}\\
		=& \accase{\aiinjn i{\om{M}}AB}{x_1}{\om{(P_1)}}{x_2}{\om{(P_2)}}C & \text{(by def.~of $\om{(\cdot)}$)}\\
		\to^+_{\beta\eta}&[\om{M}/x_i]\om{(P_i)} & \text{(by Lemma \ref{lem:optimized-admissible-beta-disjunction})}\\
		\to_{\beta}^*&\om{([M/x_i]P_i)}&\text{(by Lemma \ref{lem:optimized-translation-commutes-subst})}
	\end{array}
	$$
	
	Cases $\etai$ and $\etac$: By Lemma \ref{lem:admissible-beta-eta-implication-conj}.
	
	Case $\etad$: By Lemma \ref{lem:optimized-admissible-eta-disjunction}.
	
	Case $\pii$: Lemma \ref{lem:optimized-admissible-commutative-conversions}, item 1, just gives $\to_{\beta}^*$, not syntactic identity. For syntactic identity, see the next section of the paper.
	
	Case $\pic$: By Lemma \ref{lem:optimized-admissible-commutative-conversions}, item 2.
	
	
	Case $\abs$, $\connective=\supset,\wedge$: By Lemma \ref{lem:optimized-admissible-absurdity-conversions}, items 1. and 2. respectively.
	
	Case $\abs$, $\connective=\vee,\perp$: By Lemmas \ref{lem:optimized-admissible-absurdity-disjunction} and \ref{lem:optimized-admissible-absurdity-absurdity} respectively.
	
	Case $\pis$, $\connective=\vee,\perp$: By Lemmas \ref{lem:admissible-pi-disjunction-reduction} and \ref{lem:admissible-pi-absurdity-reduction} respectively.
	
	Let us analyze the inductive cases. In what follows $\too\in \{\to_{\beta}^*, \to_{\beta\eta}^*, =\}$.
	
	Case $\lb x.M_1\to \lb x.M_2$ with $M_1\to M_2$. By IH, $\om M_1\too \om M_2$. By Corollary \ref{cor-lem:additional-compatibility-rules-1}.5, we have
	
	$$\om {(\lb x. M_1)}=\llb x. \om M_1\too \llb x. \om M_2=\om {(\lb x. M_2)}$$. 
	
	Case $M_1N\to M_2N$ with $M_1\to M_2$. By IH, $\om M_1\too \om M_2$. By Corollary \ref{cor-lem:additional-compatibility-rules-1}.1, we have
	
	$$\om {(M_1N)}=\om M_1@\om N\too \om M_2@\om N=\om {(M_2N)}$$. 
	
	Case $MN_1\to MN_2$ with $N_1\to N_2$. By IH, $\om N_1\too \om N_2$. By Corollary \ref{cor-lem:additional-compatibility-rules-1}.2, we have
	
	$$\om {(MN_1)}=\om M@\om N_1\too \om M@\om N_2=\om{(MN_2)}.$$ 
	
	Case $\pair{M_1}{N}\to \pair{M_2}{N}$ with $M_1\to M_2$. By IH, $\om M_1\too \om M_2$. By Corollary \ref{cor-lem:additional-compatibility-rules-1}.3, we have
	
	$$\om {\pair{M_1}{N}}=\ppair{\om M_1}{\om N}\too \ppair{\om M_2}{\om N}=\om{\pair{M_2}{N}}.$$
	
	Case $\pair{M}{N_1}\to \pair{M}{N_2}$ with $N_1\to N_2$. Analogous applying Corollary \ref{cor-lem:additional-compatibility-rules-1}.4.
	
	Case $M_1i\to M_2i$ with $M_1\to M_2$. By IH, $\om M_1\too \om M_2$. By Corollary \ref{cor-lem:additional-compatibility-rules-1}.1, we have
	
	$$\om {(M_1i)}=\om M_1@i\too \om M_2@i=\om{(M_2i)}.$$ 
	
	Case $\injn i{M_1}AB\to \injn i{M_2}AB$ with $M_1\to M_2$. By IH, $\om M_1\too \om M_2$. By Lemma  \ref{lem:additional-compatibility-rules-2}.1, we have
	
	$$\om {\injn i{M_1}AB}= \aiinjn i{\om M_1}{A\om }{\om B}\too \aiinjn i{\om M_2}{A\om }{\om B}=\om {\injn i{M_2}AB}.$$ 
	
	Case $\case {M_1}{x}P{y}QC\to \case {M_2}{x}P{y}QC$ with $M_1\to M_2$. By IH, $\om M_1\too \om M_2$. By Lemma  \ref{lem:additional-compatibility-rules-2}.2, we have
	
	$$\om {LSH}= \accase {\om{M_1}}{x}{\om P}{y}{\om Q}{\om C}\too \accase {\om M_2}{x}{\om P}{y}{\om Q}{\om C}=\om {RHS}.$$ 
	
	Case $\case {M}{x}{P_1}{y}QC\to \case {M}{x}{P_2}{y}QC$ with $P_1\to P_2$. By IH, $\om P_1\too \om P_2$. By Lemma  \ref{lem:additional-compatibility-rules-2}.3, we have
	
	$$\om {LSH}= \accase {\om{M}}{x}{\om P_1}{y}{\om Q}{\om C}\too \accase {\om M}{x}{\om P_2}{y}{\om Q}{\om C}=\om {RHS}.$$ 
	
	Case $\case {M}{x}{P}{y}{Q_1}C\to \case {M}{x}{P}{y}{Q_2}C$ with $Q_1\to Q_2$ is analogous to the previous one.
	
	Case $\abort {M_1}C\to \abort {M_2}C $ with $M_1\to M_2$. By IH, $\om M_1\too \om M_2$. By Lemma  \ref{lem:additional-compatibility-rules-2}.5, we have
	
	$$\om {LSH}= \aabbort {\om M_1}{\om C}\too  \aabbort {\om M_2}{\om C}=\om {RHS}.$$ 
\end{proof}

\section{Special treatment of $\pii$}\label{sec:special}

We give 
in this separate section the proof of the base case $R=\pii$ of Theorem \ref{thm:simulation}. With the exception of $\pii$ (recall item 1 of Lemma \ref{lem:optimized-admissible-commutative-conversions}), all the commuting conversion rules of $\ipc$ hold in $\fat$ as admissible syntactic equalities - thus involving arbitrary terms and formulas of $\fat$. But, in fact, for the identification of commuting conversions in the image of the translation $\om{(\cdot)}$, it is sufficient that such admissible equalities hold for terms and formulas in that image.

\begin{thm}\label{theo:pii}
Let 
$$
\begin{array}{rcl}
M_1 & := & \case{M}{x_1}{P_1}{x_2}{P_2}{C\supset D}N\\
M_2 & := & \case{M}{x_1}{P_1 N}{x_2}{P_2 N}{D}
\end{array}
$$
Then
$$
\begin{array}{rcl}
	\om{M_1}&=&\accase{\om M}{x_1}{\om {P_1}}{x_2}{\om{P_2}}{\om C\supset \om D}@\om N\\
	&=&\accase{\om M}{x_1}{\om {P_1}@\om N}{x_2}{\om{P_2}@\om N}{\om D}\\
	&=&\om {M_2}
\end{array}
$$
In particular, item 1 of Lemma \ref{lem:optimized-admissible-commutative-conversions} holds as equality, when all the terms and formulas involved are in the image $\om{(\_)}$, and the base case $R=\pii$ of Theorem \ref{thm:simulation} holds.
\end{thm}

First we need some definitions and auxiliary results.

\begin{con}
	In what follows:
	\begin{itemize}
	\item[-] $U$ stands for a term, a type variable or a projection symbol.
	\item[-] $\vec{U}$ denotes a list $U_1,\ldots, U_m$, with $m\geq 0$.
	\item[-] If~ $\vec{U}=U_1,\ldots, U_m$ and $M$ is a term in $\fat$ then $M@\vec{U}$ denotes the term $M@U_1@\ldots @U_m$.	
	\end{itemize}
\end{con}

\begin{defn} Let $P$ be a term in $\fat$ and $z$ a term variable: $P$ is \emph{$z$-special} if, for all $\vec{U}$, $P@U\neq z$.
\end{defn}

\begin{defn} Let $P$ be a term in $\fat$:
	\begin{enumerate}
		\item $P$ is \emph{var-special} if, for all $z\notin \FV P$, $P$ is $z$-special. 
		Or equivalently: $P$ is var-special if, for all term variable $z$, and all $\vec{U}$,
		$$P@\vec{U}=z\implies z\in \FV P \enspace.$$
		
		\item $P$ is \emph{pair-special} if, for all terms $M_1$, $M_2$, and all $\vec{U}$, 
		
		$$P@\vec{U}=\pair{M_1}{M_2}\implies \FV{M_1}=\FV{M_2}.$$

	\end{enumerate}
\end{defn}


\begin{lem} \label{lem:special properties}
	\begin{enumerate} In $\fat$, let $P$, $N$ be terms and $Y$ a type variable:
		\item If $P$ is $z$-special then $P@U$ is $z$-special.
		\item If $P$ is var-special then $P@N$ and $P@Y$ are var-special.
		\item If $P$ is var-special and $P$ is pair-special then $P@i$ is var-special.
		\item If $P$ is pair-special then $P@U$ is pair-special.
	\end{enumerate}
\end{lem}

\begin{proof}
	Item 1.\ is immediate by definition of $z$-special.
	
	Item 2.\ is immediate since $\FV{P}\subseteq \FV{P@N}$ and $\FV{P}\subseteq \FV{P@Y}$.
	
	For item 3., take $P@i@\vec{U}=w$. Since $P$ is var-special, we know that $w\in \FV{P}$. If $P$ is not a pair then $P@i=Pi$ and $w\in \FV{Pi}=\FV{P}$. If $P$ is a pair, say $P=\pair{M_1}{M_2}$ then $P@i=M_i$. Since $P$ is pair-special we know that $\FV{M_1}=\FV{M_2}=\FV{P}$, thus $w\in \FV{M_i}$.
	
	Item 4.\ is immediate by definition of pair-special.  
\end{proof}

An important remark: If $P$ is pair special, then $\FV{P@1}=\FV{P}=\FV{P@2}$.

\begin{lem}\label{lem:z-special}
	Suppose $P,Q$ are $z$-special and $z\notin M$. Then
	$$[N/z]\accase{M}{x}{P}{y}{Q}{D}=\accase{M}{x}{[N/z]P}{y}{[N/z]Q}{D}$$	
\end{lem}

\begin{proof}
	The proof is by induction on $D$.
	
	Case $D=X$.
$$
\begin{array}{rcll}
	LHS&=&[N/z]((M@X)@\pair{\lb x.P}{\lb y.Q})&
	\\
	&=&([N/z](M@X))@[N/z]\pair{\lb x.P}{\lb y.Q}&
	\\
		&=&(([N/z]M)@X)@\pair{\lb x.[N/z]P}{\lb y.[N/z]Q}&
		\\
	&=&RHS&
\end{array}
$$
	
The first equality above follows from the definition of $\accasesymb$; the second uses item 1. of Lemma \ref{lem:term-subst-@} since $M@X\neq z$; the third uses the same item noticing that $M\neq z$; the fourth equality uses the fact that $z\notin M$.	

	Case $D=D_1\supset D_2$.
$$
\begin{array}{rcll}
	LHS&=&[N/z](\lb w.\accase{M}{x}{P@w}{y}{Q@w}{D_2})&\textrm{(def. of $\accasesymb$)}\\
	&=&\lb w.[N/z]\accase{M}{x}{P@w}{y}{Q@w}{D_2}&\\
	&=&\lb w.\accase{M}{x}{[N/z](P@w)}{y}{[N/z](Q@w)}{D_2}&
	\\
	&=&\lb w.\accase{M}{x}{([N/z]P)@w}{y}{([N/z]Q)@w}{D_2}&
	\\
	&=&RHS&
\end{array}
$$

The third equality above follows from IH using item 1.\ of Lemma \ref{lem:special properties} and the fourth equality uses item 1.\ of Lemma \ref{lem:term-subst-@} since $P$ and $Q$ are $z$-special.

Case $D=D_1\wedge D_2$.
$$
\begin{array}{rcll}
	LHS&=&[N/z]\ipair i{\accase Mx{P@i}y{Q@i}{D_i}}&\textrm{(def. of $\accasesymb$)}\\
	&=&\ipair i{[N/z]\accase Mx{P@i}y{Q@i}{D_i}}&\\
	&=&\ipair i{\accase Mx{[N/z](P@i)}y{[N/z](Q@i)}{D_i}}&
	\\
	&=&\ipair i{\accase Mx{([N/z]P)@i}y{([N/z]Q)@i}{D_i}}&
	\\
	&=&RHS&
\end{array}
$$

The third equality above follows from IH using item 1.\ of Lemma \ref{lem:special properties} and the fourth equality uses item 1.\ of Lemma \ref{lem:term-subst-@} since $P$ and $Q$ are $z$-special.

Case $D=\forall X. D_0$.
$$
\begin{array}{rcll}
	LHS&=&[N/z](\Lb X.\accase{M}{x}{P@X}{y}{Q@X}{D_0})&\textrm{(def. of $\accasesymb$)}\\
&=&\Lb X.[N/z]\accase{M}{x}{P@X}{y}{Q@X}{D_0}&\\
&=&\lb X.\accase{M}{x}{[N/z](P@X)}{y}{[N/z](Q@X)}{D_0}&
\\
&=&\lb X.\accase{M}{x}{([N/z]P)@X}{y}{([N/z]Q)@X}{D_0}&
\\
&=&RHS&
\end{array}
$$

The third equality above follows from IH using item 1.\ of Lemma \ref{lem:special properties} and the fourth equality uses item 1.\ of Lemma \ref{lem:term-subst-@} since $P$ and $Q$ are $z$-special. 
\end{proof}

The previous lemma should be compared with items 1(a) and 3(a) of Lemma \ref{lem:subst-admissible-optimized-constructions}\footnote{For instance, the counter-example in footnote \ref{fn:counter-example} is not a counter-example to the previous lemma because $\lb w.w$ is not $z$-special.}. The third item of the next lemma should be compared with Lemma \ref{lem:FV_case}. In its statement, we say $M$ is \emph{specific} if $M$ has the form $\Lambda Y.\lb x.M'$ with $x\notin\FV{M'}$.
\begin{lem}\label{lem:miracle_case}
	For all terms $M,P,Q$ and types $C$ in $\fat$, if $M,P,Q$ are var-special and pair-special then
	\begin{enumerate}
		\item $\accase MxPyQC$ is var-special
		\item $\accase MxPyQC$ is pair-special
		\item If $M$ is not specific then 
		$$\FV{\accase MxPyQC}=\FV M\cup (\FV P\setminus x)\cup (\FV Q\setminus y),$$
		else $\FV{\accase MxPyQC}=\FV{M}$.
	\end{enumerate}
\end{lem}

\begin{proof}
	Because of Lemma \ref{lem:FV_case}, instead of proving item 1, it is enough to prove that for all $w$, for all $\vec{U}$,
	$$\accase MxPyQC@\vec{U}=w\Rightarrow w\in \FV M.$$ In the remainder of this proof, we refer to this variant as item 1. 
	
	The conjunction of the three items is proved by induction on $C$.
	
	Case $C=X$. Then $\accase MxPyQX=M@X@\pair{\lb x. P}{\lb y.Q}$.
	
	Item 1. Let $\accase MxPyQX@\vec{U}=M@X@\pair{\lb x. P} {\lb y.Q}@\vec{U}=w$. Since $M$ is var-special, $w\in \FV{M}$.
	
	Item 2. Let $\accase MxPyQX@\vec{U}=M@X@\pair{\lb x. P} {\lb y.Q}@\vec{U}=\pair{M_1}{M_2}$. Since $M$ is pair-special, $\FV{M_1}=\FV{M_2}$.
	
	Item 3. Let us analyze the free variables of $\accase MxPyQX$, i.e., the free variables of $M@X@\pair{\lb x. P} {\lb y.Q}$.
	
	If $M\neq \Lb Y.M_0$ then $\accase MxPyQX=MX\pair{\lb x.P}{\lb y.Q}$. Thus $\FV{\accase MxPyQC}=\FV M\cup (\FV P\setminus x)\cup (\FV Q\setminus y)$.
	
	If $M=\Lb Y.M_0$ and $M_0\neq \lb z.M_1$ then  $\accase MxPyQX$ is\\ $([X/Y]M_0)\pair{\lb x.P}{\lb y.Q}$. Thus $\FV{\accase MxPyQC}=\FV {[X/Y]M_0}\cup (\FV P\setminus x)\cup (\FV Q\setminus y)= \FV {M}\cup (\FV P\setminus x)\cup (\FV Q\setminus y)$. The latter equality is justified by the fact that $\FV{[X/Y]M_0}=\FV{M_0}=\FV{M}$.
	
	If $M=\Lb Y.\lb z.M_1$, with $z\in\FV{M_1}$,  then 
	$$
	\begin{array}{rcll}
	\FV{\accase MxPyQX}&=&\FV{[\pair{\lb x.P}{\lb y.Q}/z][X/Y]M_1}&\\
	&=&\left(\FV{[X/Y]M_1}\setminus z\right)\cup\FV{\pair{\lb x.P}{\lb y.Q}}&\\
	&=&\left(\FV{M_1}\setminus z\right)\cup(\FV{P}\setminus x)\cup (\FV{Q}\setminus y)&\\
	\end{array}
	$$
	In the second equation we used the fact that $z\in\FV{M_1}$.
	
	If $M=\Lambda Y.\lb z.M_1$, with $z\notin\FV{M_1}$, that is, $M$ is specific. Then 
	$$
	\begin{array}{rcl}
	\FV{\accase MxPyQX}&=&\FV{[X/Y]M_1}\\
	&=&\FV{M_1}\\
	&=&\FV{M}
	\end{array}
	$$ 
	In this calculation, the fact $z\notin\FV{M_1}$ is used in the first and last equations.
	
	Case $C=C_1\wedge C_2$. Then $\accase MxPyQC@\vec{U}=\pair{N_1}{N_2}$ with $N_i=\accase{M}{x}{P@i}{y}{Q@i}{C_i}$, for $i=1,2$. Before proving the items, we want to prove that $\FV{N_1}=\FV{N_2}$. If $M$ is not specific, then:
	
	$$
	\begin{array}{rcll}
		\FV{N_1}&=&\FV{\accase Mx{P@1}y{Q@1}{C_1}} 	&\\
		&=&\FV{M}\cup (\FV{P@1}\setminus x) \cup (\FV{Q@1}\setminus y)&\text{(by IH)}\\
		&=&\FV{M}\cup (\FV{P@2}\setminus x) \cup (\FV{Q@2}\setminus y)&\text{(*)}\\
		&=&\FV{\accase Mx{P@2}y{Q@2}{C_2}}&\textrm{(by IH)}\\	
		&=&\FV{N_2}&
	\end{array}
	$$
	Justification $(*)$: Since $P$, $Q$ are pair-special, then $\FV{P@1}=\FV{P@2}$ and $\FV{Q@1}=\FV{Q@2}$. If $M$ is specific, then $\FV{N_1}=\FV{M}=\FV{N_2}$, with the two equalities justified by IH.
	
		Item 1. Let $\accase MxPyQC@\vec{U}=
		\ipair i{N_i}@U_1@\ldots @U_n=w$. The cases $n=0$, or $n\geq 1$ and $U_1\neq i$, are impossible. So assume $n\geq 1$ and $U_1=i$.  Then 	$\accase MxPyQC@\vec{U}=\accase M {x}{P@i}{y}{Q@i}{C_i}@U_2@\ldots@U_n=w.$	By Lemma \ref{lem:special properties}, since $P$ and $Q$ are var-special and pair-special we know that $P@i$ and $Q@i$ are var-special and pair-special. So IH applies, and we obtain $w\in \FV{M}$.
		
	Item 2. Let $\accase MxPyQC@\vec{U}=\ipair i{N_i}@U_1@\ldots @U_n=\pair{M_1}{M_2}$. The case $n\geq 1$ and $U_1\neq i$ is impossible. 
		
		If $n\geq 1$ and $U_1=i$ we have $\accase Mx{P@i}y{Q@i}{C_i}@U_2@\ldots@U_n=\pair{M_1}{M_2}$. Again, since $P@i$ and $Q@i$ are var-special and pair-special, by IH we have $\FV{M_1}=\FV{M_2}$.
		
		If $n=0$ we have $M_i=N_i$, for each $i=1,2$. We already saw that $\FV{N_1}=\FV{N_2}$. 

	Item 3. Suppose $M$ is not specific. Choose $j\in\{1,2\}$. Then:
		$$
	\begin{array}{cll}
		&\FV{\accase MxPyQC}\\
		=&\FV{\ipair i{N_i}}&\textrm{(by def. of $\accasesymb$)}\\
		=&\bigcup_{i=1,2}\FV{N_i}&\\
		=&\FV{\accase{M}{x}{P@j}{y}{Q@j}{C_j}}&\textrm{(since $\FV{N_1}=\FV{N_2}$)}\\
		=&\FV M\cup (\FV {P@j}\setminus x)\cup (\FV {Q@j}\setminus y)&\text{(by IH)}\\
		=&\FV M\cup (\FV {P}\setminus x)\cup (\FV {Q}\setminus y)&\textrm{(*)}
	\end{array}
	$$
		Justification $(*)$: Since $P$, $Q$ are pair-special, then $\FV{P@j}=\FV{P}$ and $\FV{Q@j}=\FV{Q}$. Finally, suppose $M$ is specific. Again, for some $j$, \\$\FV{\accase MxPyQC}=\FV{N_j}$. But, in this case, IH gives $\FV{N_j}=\FV{M}$, as required.
	
		Case $C=C_1\supset C_2$. \\Then $\accase MxPyQC=(\lb z^{C_1}.\accase{M}{x}{P@z}{y}{Q@z}{C_2})$.
	
	Item 1. Let 
	$$
	\begin{array}{cl}
		 &\accase MxPyQC@\vec{U}\\
		=&(\lb z^{C_1}.\accase{M}{x}{P@z}{y}{Q@z}{C_2})@U_1@\ldots @U_n\\
		=&w
	\end{array}
	$$
	 The cases $n=0$, or $n\geq 1$ and $U_1\neq N$, are impossible. So suppose $n\geq 1$ and $U_1=N$. We have $([N/z]\accase{M}{x}{P@z}{y}{Q@z}{C_2})@U_2@\ldots@U_n=w$. Note that $z\notin M,P,Q$, and by being var-special, $P$, $Q$ are $z$-special. Thus by Lemma \ref{lem:z-special} we have 
	$$
	\begin{array}{cl}
		&([N/z]\accase{M}{x}{P@z}{y}{Q@z}{C_2})@U_2@\ldots@U_n\\
		=&\accase{M}{x}{P@N}{y}{Q@N}{C_2}@U_2@\ldots@U_n\\
		=&w
	\end{array}
    $$
   Therefore, by IH, using Lemma \ref{lem:special properties}, we obtain $w\in \FV{M}$.
	 
	 	Item 2. Let 
	 	$$
	 	\begin{array}{cl}
	 	&\accase MxPyQC@\vec{U}\\
	 	=&(\lb z^{C_1}.\accase{M}{x}{P@z}{y}{Q@z}{C_2})@U_1@\ldots @U_n\\
	 	=&\pair{M_1}{M_2}
	 	\end{array}
	 	$$
	 	The cases $n=0$, or $n\geq 1$ and $U_1\neq N$, are impossible. So suppose $n\geq 1$ and $U_1=N$. We have $([N/z]\accase{M}{x}{P@z}{y}{Q@z}{C_2})@U_2@\ldots@U_n=\pair{M_1}{M_2}$. As before, 	$$
	 	\begin{array}{cl}
	 	&([N/z]\accase{M}{x}{P@z}{y}{Q@z}{C_2})@U_2@\ldots@U_n\\
	 	=&\accase{M}{x}{P@N}{y}{Q@N}{C_2})@U_2@\ldots@U_n\\
	 	=&\pair {M_1}{M_2}
	 	\end{array}
	 	$$  
	 	Thus, by IH, we have $\FV{M_1}=\FV{M_2}$.
	 	
	 	Item 3. Suppose $M$ is not specific. Then
		$$
	\begin{array}{cll}
		&\FV{\accase MxPyQC}\\
		=&\FV{\lb z. \accase{M}{x}{P@z}{y}{Q@z}{C_2}}&\textrm{(def. of $\accasesymb$)}\\
		=&\FV{\accase{M}{x}{P@z}{y}{Q@z}{C_2}}\setminus z&\\
		=&(\FV M \cup (\FV{P@z}\setminus x)\cup (\FV{Q@z}\setminus y))\setminus z&\textrm{(IH using Lemma \ref{lem:special properties})}\\
		=&\FV M \cup (\FV{P@z}\setminus x\setminus z)\cup (\FV{Q@z}\setminus y\setminus z)&\textrm{($z\notin \FV M$)}\\
		=&\FV M \cup (\FV{P@z}\setminus z\setminus x)\cup (\FV{Q@z}\setminus z\setminus y)&\\
		=&\FV M \cup (\FV{P}\setminus x)\cup (\FV{Q}\setminus y)&\textrm{($z\notin \FV  P,\FV{Q}$)}
	\end{array}
	$$ 
	Finally suppose $M$ is specific. Again, 
	$$\FV{\accase MxPyQC}=\FV{\accase{M}{x}{P@z}{y}{Q@z}{C_2}}\setminus z,$$
	and IH in this case says this is $\FV{M}\setminus z$. We may assume we have chosen $z\notin \FV{M}$, so we obtain $\FV{M}$, as required.
	 
	 	Case $C=\forall X. C_0$. Then $\accase MxPyQC=(\Lb X.\accase{M}{x}{P@X}{y}{Q@X}{C_0})$.
	 
	 	Item 1. Let 
	 	$$
	 	\begin{array}{cl}
	 	&\accase MxPyQC@\vec{U}\\
	 	=&(\Lb X.\accase{M}{x}{P@X}{y}{Q@X}{C_0})@U_1@\ldots @U_n\\
	 	=&w
	 	\end{array}
	 	$$	
	 	The cases $n=0$, or $n\geq 1$ and $U_1\neq Y$, are impossible. So let $n\geq 1$ and $U_1=Y$. We have $([Y/X]\accase{M}{x}{P@X}{y}{Q@X}{C_0})@U_2@\ldots@U_n=w$. But, by item 1.b of Lemma \ref{lem:subst-admissible-optimized-constructions} and Lemma \ref{lem:type-subst-@} (since $X\notin M,P,Q$), we know that 
	 	$$
	 	\begin{array}{cl}
	 		&([Y/X]\accase{M}{x}{P@X}{y}{Q@X}{C_0})@U_2@\ldots@U_n\\
	 		=&\accase{M}{x}{P@Y}{y}{Q@Y}{[Y/X]C_0}@U_2@\ldots@U_n\\
	 		=&w
	 	\end{array}	
	 	$$ 
	 	Thus, by IH using Lemma \ref{lem:special properties}, we have $w\in \FV M$.
	 	
	 	Item 2. Let	
	 	$$
	 	\begin{array}{cl}
	 	&\accase MxPyQC@\vec{U}\\
	 	=&(\Lb X.\accase{M}{x}{P@X}{y}{Q@X}{C_0})@U_1@\ldots @U_n\\
	 	=&\pair{M_1}{M_2}
	 	\end{array}
	 	$$	
	 	The cases $n=0$, or $n\geq 1$ and $U_1\neq Y$, are impossible. So suppose $n\geq 1$ and $U_1=Y$. We have $([Y/X]\accase{M}{x}{P@X}{y}{Q@X}{C_0})@U_2@\ldots@U_n=\pair{M_1}{M_2}$. Again, by item 1.b of Lemma \ref{lem:subst-admissible-optimized-constructions} and Lemma \ref{lem:type-subst-@} (since $X\notin M,P,Q$), we know that 
	 
	  $$
	  \begin{array}{cl}
	  &([Y/X]\accase{M}{x}{P@X}{y}{Q@X}{C_0})@U_2@\ldots@U_n\\
	  =&(\accase{M}{x}{P@Y}{y}{Q@Y}{[Y/X]C_0})@U_2@\ldots@U_n\\
	  =&\pair{M_1}{M_2}
	  \end{array}
	  $$
	  Thus, by IH, and using Lemma \ref{lem:special properties}, we have $\FV{M_1}=\FV{M_2}$.
	  
	  Item 3. Suppose $M$ is not specific. Then
	 
	 	$$
	 \begin{array}{cll}
	 	&\FV{\accase MxPyQC}\\
	 	=&\FV{\Lb X. \accase{M}{x}{P@X}{y}{Q@X}{C_0}}&\textrm{(by def. of $\accasesymb$)}\\
	 	=&\FV{\accase{M}{x}{P@X}{y}{Q@X}{C_0}}&\\
	 	=&\FV M \cup (\FV{P@X}\setminus x)\cup (\FV{Q@X}\setminus y)&\textrm{(by IH using Lemma \ref{lem:special properties})}\\
	 	=&\FV M \cup (\FV{P}\setminus x)\cup (\FV{Q}\setminus y)&\textrm{(by Lemma \ref{lem:FV_om_@})}
	 \end{array}
	 $$ 
	 Finally, suppose $M$ is specific. Then 
	 $$\FV{\accase MxPyQC}=\FV{\accase{M}{x}{P@X}{y}{Q@X}{C_0}}=\FV{M},$$
	 with the last equation given by IH.
\end{proof}

Now we move to $\aabbortsymb$. While for $\accasesymb$ the three statements of Lemma \ref{lem:miracle_case} had to be proved together, the similar statements for $\aabbortsymb$ can be proved separately. One of them was already given as Lemma \ref{lem:miracle_abort_FV}. The other two are the next two lemmas.

\begin{lem}\label{lem:miracle_abort_var}
	For all terms $M$ and types $A$ in $\fat$
	
	$$M \textup{~ var-special~} \implies \aabbort{M}{A} \textup{~ var-special}. $$
\end{lem}

\begin{proof}
	Suppose $M$ var-special. Due to Lemma \ref{lem:miracle_abort_FV}, it suffices to prove that for all term variable $w$, and all $\vec{U}=U_1,\cdots,U_n$,
	
	$$\aabbort M A @\vec{U}=w \Rightarrow w\in \FV M.$$
	
	The proof is by induction on $A$.
	
	Case $A=X$. Let $\aabbort M X @\vec{U}=M@X@\vec{U}=w$. Since $M$ is var-special, $w\in \FV{M}$.
	
	Case $A=A_1\wedge A_2$. Let $\aabbort MA @\vec{U}=\ipair i{\aabbort{M}{A_i}}@\vec{U}=w$. The case $n=0$ is impossible, since a pair is not a variable. The case $n\geq 1$ and $U_1\neq i$ is impossible, since an application is not a variable. So let $n\geq 1$ and $U_1=i$. We have $\aabbort M{A_i}@U_2@\ldots@U_n=w$. By IH, $w\in \FV{M}$.	
	
		Case $A=A_1\supset A_2$. Let $\aabbort MA @\vec{U}=(\lb z^{A_1}.\aabbort M {A_2})@\vec{U}=w$. Again, the cases $n=0$, or $n\geq 1$ and $U_1\neq N$ are impossible. So let $n\geq 1$ and $U_1=N$. We have $([N/z]\aabbort M {A_2})@U_2@\ldots@U_n=w$. But, by item 3.b of Lemma \ref{lem:subst-admissible-optimized-constructions} (since $z\notin M$), we know that $[N/z]\aabbort M {A_2}=\aabbort M{A_2}$. Hence $(\aabbort {M} {A_2})@U_2@\ldots@U_n=w$. By IH, we have $w\in \FV M$.
	
	Case $A=\forall X.A_0$. Let $\aabbort MA @\vec{U}=(\Lb X.\aabbort M {A_0})@\vec{U}=w$. Again, the cases $n=0$, or $n\geq 1$ and $U_1\neq Y$ are impossible. So let $n\geq 1$ and $U_1=Y$. We have $([Y/X]\aabbort M {A_0})@U_2@\ldots@U_n=w$. But, by item 2.b of Lemma \ref{lem:subst-admissible-optimized-constructions} (since $X\notin M$), we know that $[Y/X]\aabbort M {A_0}=\aabbort M{A_0}$. Hence $(\aabbort {M} {[Y/X]A_0})@U_2@\ldots@U_n=w$. By IH, we have $w\in \FV M$.
\end{proof}

\begin{lem}\label{lem:miracle_abort_pair}
	For all terms $M$ and types $A$ in $\fat$,

$$M \textup{~ pair-special~} \implies \aabbort{M}{A} \textup{~ pair-special}. $$
\end{lem}

\begin{proof}
	Suppose $M$ is pair-special. We want to show that, for all $M_1,M_2$, for all $\vec{U}=U_1,\cdots,U_n$, 
	$$\aabbort{M}{A}@\vec{U}=\pair{M_1}{M_2} \implies FV(M_1)=FV(M_2).$$
	
	The proof is by induction on $A$.
	
	Case $A=X$. Let $\aabbort M X @\vec{U}=M@X@\vec{U}=\pair{M_1}{M_2}$. Since $M$ is pair-special, $\FV{M_1}=\FV{M_2}$.
	
	Case $A=A_1\wedge A_2$. Let $\aabbort MA @\vec{U}=\ipair i{\aabbort{M}{A_i}}@\vec{U}=\pair{M_1}{M_2}$. If $n=0$ we have $M_1=\aabbort M {A_1}$ and $M_2=\aabbort M {A_2}$. By Lemma \ref{lem:miracle_abort_FV} we have $\FV{M_1}=\FV{M}=\FV{M_2}$. If $n\geq 1$, the case $U_1\neq i$ is impossible, because an application is not a pair. So let $U_1=i$. We have $\aabbort M{A_i}@U_2@\ldots@U_n=\pair{M_1}{M_2}$. By IH, $\FV{M_1}=\FV{M_2}$.
	
	Case $A=A_1\supset A_2$. Let $\aabbort MA @\vec{U}=(\lb z^{A_1}.\aabbort M {A_2})@\vec{U}=\pair {M_1}{M_2}$. The case $n=0$ is impossible, because an abstraction is not a pair. The case $n\geq 1$ and $U_1\neq N$ is impossible, because an application is not a pair. So let $n\geq 1$ and $U_1=N$. We have $([N/z]\aabbort M {A_2})@U_2@\ldots@U_n=\pair{M_1}{M_2}$. But, by item 3.b of Lemma \ref{lem:subst-admissible-optimized-constructions} (since $z\notin M$), we have $([N/z]\aabbort M {A_2})=\aabbort{M}{A_2}$. Hence $(\aabbort {M} {A_2})@U_2@\ldots@U_n=\pair{M_1}{M_2}$. Thus, by IH, we have $\FV{M_1}=\FV{M_2}$.
	
	Case $A=\forall X.A_0$. Let $\aabbort MA @\vec{U}=(\Lb X.\aabbort M {A_0})@\vec{U}=\pair{M_1}{M_2}$. Again, the cases $n=0$, or $n\geq 1$ and $U_1\neq Y$ are impossible. So let $n\geq 1$ and $U_1=Y$. We have $([Y/X]\aabbort M {A_0})@U_2@\ldots@U_n=\pair{M_1}{M_2}$. But, by item 2.b of Lemma \ref{lem:subst-admissible-optimized-constructions} (since $X\notin M$), we have $[Y/X]\aabbort M {A_0}=\aabbort{M}{[Y/X]A_0}$. Hence $(\aabbort {M} {[Y/X]A_0})@U_2@\ldots@U_n=\pair{M_1}{M_2}$. Thus, by IH, we have $\FV{M_1}=\FV{M_2}$.
\end{proof}

\begin{lem}\label{lem:om_var_pair_special}
For all terms $P$ in $\ipc$, $\om P$ is both var-special and pair-special.
\end{lem}

\begin{proof}
	The proof is by induction on $P$. Let $w$ be an assumption variable, and $M_1, M_2$ be terms in $\fat$. Let $\vec{U}=U_1,\ldots ,U_n$, with $n\geq 0$.
	
	Case $P=x$. 
	
	Assume that $\om P@\vec{U}= w$. That is, $x@U_1@\ldots @U_n= w$. But then $n=0$ and $x=w$. Therefore $w\in \FV x$. Therefore, $\om P$ is var-special.
	
	Assume that $\om P@\vec{U}= \pair{M_1}{M_2}$. Impossible since, $x@U_1@\ldots @U_n$ is never a pair. Therefore , $\om P$ is pair-special.
	
	Case $P=\lb x. P_0$. Then $\om P@\vec{U}=
	 (\lb w.(\lb x.\om {P_0})w)@U_1@\ldots @U_n$.
	
	If $n=0$ then $\om P@\vec{U}$ is neither a variable, nor a pair.
	
	If $n\geq 1$ and $U_1\neq N$ then $\om P@\vec{U}$ is neither a variable, nor a pair.
	
	If $n\geq 1$ and $U_1=N$, then $\om P@\vec{U}=((\lb x. \om {P_0})N)@U_2@\ldots @U_n=$\\$=(\lb x. \om {P_0})NU_2\ldots U_n$ which is neither a variable, nor a pair.
	
		Case $P=\pair {P_1}{P_2}$. Then $\om P@\vec{U}=
		 \pair{M 1}{M 2}@U_1@\ldots @U_n$, with $M=\pair{\om {P_1}}{\om {P_2}}$.
	
	If $n=0$ then $\om P@\vec{U}$ is not a variable, it is a pair the pair $\pair{M 1}{M 2}$, and $FV(M 1)=FV(M)=FV(M 2)$.
	
	If $n\geq 1$ and $U_1\neq i$ then $\om P@\vec{U}$ is neither a variable, nor a pair.
	
	If $n\geq 1$ and $U_1=i$, then $\om P@\vec{U}=(Mi)@U_2@\ldots@ U_n=MiU_2\ldots U_n$ which is neither a variable, nor a pair.
	
		Case $P=\injn i{P_0}{A}{B}$. Then $\om P@\vec{U}=
		 (\Lb X. \lb w. wi\om{P_0})@U_1@\ldots @U_n$, with $X,w\notin\om{P_0}$.
	
	If $n=0$ or $n= 1$ then $\om P@\vec{U}$ is neither a variable, nor a pair. 
	
	If $n\geq 2$ but not $U_1= Y$ and $U_2=N$ simultaneously then $\om P@\vec{U}$ is neither a variable, nor a pair.
	
	If $n\geq 2$ and $U_1=Y$ and $U_2=N$, then $\om P@\vec{U}=(Ni\om{P_0})@U_3@\ldots@ U_n=Ni\om {P_0}U_3\ldots U_n$ which is neither a variable, nor a pair.
	
		Case $P=P_0Q_0$. Then $\om P@\vec{U}= \om {P_0}@\om {Q_0}@U_1@\ldots @U_n$. 
	
	Assume $\om P@\vec{U}=w$. That is, $\om {P_0}@\om {Q_0}@U_1@\ldots @U_n=w$. By IH, $\om {P_0}$ is var-special, thus $w\in \FV{\om {P_0}}$. Therefore, by Lemma \ref{lem:FV_om_@}, $w\in \FV{\om {P_0}@\om {Q_0}}$. 
	 
	Assume $\om P@\vec{U}=\pair{M_1}{M_2}$. That is, $\om {P_0}@\om {Q_0}@U_1@\ldots @U_n=\pair{M_1}{M_2}$. By IH,$\om{P_0}$ is pair-special, hence $\FV{M_1}=\FV{M_2}$. 
	
		Case $P=P_0i$. Then $\om P@\vec{U}= \om {P_0}@i@U_1@\ldots @U_n$.
	
	Assume $\om P@\vec{U}=w$. That is, $\om {P_0}@i@U_1@\ldots @U_n=w$. By IH, $\om {P_0}$ is var-special, thus $w\in \FV{\om {P_0}}$. Let us prove that $w\in \FV{\om {P_0}@i}$. If $\om {P_0}$ is not a pair, then $\om {P_0}@i=\om {P_0}i$ and the result is immediate. If it is a pair, say $\om {P_0}=\pair{M_1}{M_2}$, then $\om {P_0}@i=M_i$. Since by IH $\om {P_0}$ is pair-special, we know that $\FV{M_1}=\FV{M_2}$, hence $\FV{\om{P_0}}=FV(M_i)$, thus $w\in \FV{\om {P_0}@i}$. 
	
	Assume $\om P@\vec{U}=\pair{M_1}{M_2}$. That is, $\om {P_0}@i@U_1@\ldots @U_n=\pair{M_1}{M_2}$. By IH, $\om{P_0}$ is pair-special, hence $\FV{M_1}=\FV{M_2}$. 
	
	Case $P=\case{M_0}{x}{P_0}{y}{Q_0}{C}$. By IH, $\om {M_0}$, $\om {P_0}$ and $\om {Q_0}$ are var-special and pair-special. Thus, by Lemma \ref{lem:miracle_case}, we know that $\accase{\om {M_0}}{x}{\om {P_0}}{y}{\om {Q_0}}{\om C}$ is var-special and pair-special. Hence $\om P$ is var-special and pair-special.
	
	Case $P=\abort{P_0}{A}$. By IH, we know that $\om {P_0}$ is var-special and pair-special. Thus, by Lemmas \ref{lem:miracle_abort_var} and \ref{lem:miracle_abort_pair}, we know that $\om P=\aabbort{\om {P_0}}{\om A}$ is var-special and pair-special.
\end{proof}

Finally, we are able to conclude:

\begin{proof}[Proof of Theorem \ref{theo:pii}]
	

$$
\begin{array}{rcll}
	\om{M_1}&=&\accase{\om M}{x_1}{\om {P_1}}{x_2}{\om{P_2}}{\om C\supset \om D}@\om N&\textrm{(def. of $\om{(\cdot)}$)}\\
	&=&(\lb z^{\om C}.\accase{\om M}{x_1}{\om {P_1}@z}{x_2}{\om{P_2}@z}{\om D}@\om N&\textrm{(def. of $\accasesymb$)}\\
	&=&[\om N/z]\accase{\om M}{x_1}{\om {P_1}@z}{x_2}{\om{P_2}@z}{\om D}&\textrm{(def. of $@$)}\\
	&=&\accase{\om M}{x_1}{[\om N/z](\om {P_1}@z)}{x_2}{[\om N/z](\om{P_2}@z)}{\om D}&\textrm{(*)}\\
	&=&\accase{\om M}{x_1}{\om {P_1}@\om N}{x_2}{\om{P_2}@\om N}{\om D}&\textrm{(**)}\\
	&=&\om {M_2}& \textrm{(def. of $\om{(\cdot)}$)}
\end{array}
$$

Justification $(*)$. The equality is by Lemma \ref{lem:z-special}. Note that, by definition of $\accasesymb$, $z\notin \om M, \om {P_1}, \om {P_2}$. By Lemma \ref{lem:om_var_pair_special}, we know that $\om{P_1}, \om {P_2}$ are var-special, hence $\om{P_1}, \om {P_2}$ are $z$-special. Hence, by item 1.\ of Lemma \ref{lem:special properties}, $\om{P_1}@z, \om {P_2}@z$ are $z$-special. So the conditions for the application of Lemma \ref{lem:z-special} are satisfied.

Justification $(**)$. From $z\notin \om {P_i}$, it follows $\om{P_i}\neq z$ and $[\om N/z]\om {P_i}=\om{P_i}$, for each $i=1,2$. The equality follows by item 1.\ of Lemma \ref{lem:term-subst-@}.
\end{proof}
	
\section{Discussion}\label{sec:discussion}

The simulation theorem (Theorem \ref{thm:simulation}) is stated with $\om{M}\too\om{N}$, where $\too$ is either $\to_{\beta}^*$, $\to_{\beta\eta}^*$, or $=$. From the statement alone, the possibility exists that the translation collapsed all the reduction steps of the source. But the inspection of the proof quickly shows this is not the case. Already the base cases of $\beta$-reduction show strict preservation of reduction at root position. The inspection of some of the inductive cases, together with the items of Corollary \ref{cor-lem:additional-compatibility-rules-1} and Lemma \ref{lem:additional-compatibility-rules-2} which speak about $\beta$-reduction, guarantees that other cases of $\beta$-reduction are also strictly preserved (i.e.\ preserved without collapse). The following corollary makes these observations more precise (the terminology ``head reduction'' used in it is coherent with the use of that designation in the theory of the untyped $\lb$-calculus).

\begin{cor}[Head $\beta$-reduction is strictly preserved] Recall in $\ipc$ $\beta=\beta_{\supset}\cup\beta_{\wedge}\cup\beta_{\vee}$. Let \emph{head} $\beta$-reduction, denoted $\to_{\beta}^{head}$, be the closure of $\beta$ under the rules in Fig.~\ref{fig:head-reduction}. If $M\to_{\beta}^{head}M'$ in $\ipc$ then $\om M \to_{\beta}^+\om N$ in $\fat$.
\end{cor}
%
%
\begin{figure}\caption{Head reduction in $\ipc$}\label{fig:head-reduction}
	$$
	\begin{array}{c}
		\infer{\lb x. M\,\to\,\lb x. M'}{M\,\to\,M'}\qquad\infer{MN\,\to\,M'N}{M\,\to\,M'}\\ \\		\infer{\pair{M}{N}\,\to\,\pair{M'}{N}}{M\,\to\,M'}\qquad\infer{\pair{M}{N}\,\to\,\pair{M}{N'}}{N\,\to\,N'}\\ \\
		\infer{\injn i{M}AB\,\to\,\injn i{M'}AB}{M\,\to\,M'}\qquad\infer{\case {M}{x}P{y}QC\,\to\,\case {M'}{x}P{y}QC}{M\,\to\,M'}\\ \\
		\infer{\abort {M}C\,\to\,\abort {M'}C}{M\,\to\,M'}
	\end{array}
	$$
\end{figure}


Notice that the fact that $\om{(\cdot)}$ collapses commuting conversions prevents larger subsets of $\ipc$'s $\beta$-reduction from being strictly simulated. Suppose $\twoheadrightarrow$ is a subset of $\to_{\beta}^*$ closed under the rule $N\twoheadrightarrow N' \implies MN\twoheadrightarrow MN'$. Suppose $N\twoheadrightarrow N'$. Then $\abort{M}{C\supset D}N\twoheadrightarrow\abort{M}{C\supset D}N'$ in $\ipc$ but $\om{(\abort{M}{C\supset D}N)}=\om{(\abort{M}{C\supset D}N')}$, because both of these terms are equal to $\om{\abort{M}{D}}$, due to the collapse of $\abi$-reduction.

The simulation theorem states that the new translation preserves reduction. Does it preserve normal forms? For trivial reasons, the answer is ``no'': just consider the way $\lb$-abstraction or pairs are translated. But there is a less trivial and more interesting example.

Consider in $\ipc$ the $\beta$-normal form $P:=\case{M}{x_1}{P_1}{x_2}{P_2}{C\supset D}$. Then 
$\om{P}=
\lb z^{\om{C}}.\accase{\om{M}}{x_1}{\om{P_1}@z}{x_2}{\om{P_1}@z}{\om{D}}$. If $P_i$ is $\lb w.Q_i$, then $\om{P_i}@z$ is a redex, namely $(\lb w.\om{Q_i})z$, occurring in $\om{P}$.

This example is not a defect specific of translation $\om{(\cdot)}$. In the translation $\am{(\cdot)}$ from \cite{EspiritoSantoFerreira2020}, the redex $(\lb w.\am{Q_i})z$ occurs in $\am{P}$ as well; and if we consider the translation $\cm{(\cdot)}$ from \cite{FerreiraFerreira2009,FerreiraFerreira2013}, already $\cm{P}$ itself is a redex, whether $P_i$ is an abstraction or not.

The defect, if we may say so, is in the concept of normal form in $\ipc$. We suggest $\ipc$ proofs should also be normalized w.~r.~t.~the rules
$$
\begin{array}{rcl}
\case{M}{x_1}{P_1}{x_2}{P_2}{C\supset D} & \to & \lb z^C.\case{M}{x_1}{P_1@z}{x_2}{P_2@z}{D}\\
\case{M}{x_1}{P_1}{x_2}{P_2}{C_1\wedge C_2} & \to & \ipair{i}{\case{M}{x_1}{{P_1}@{i}}{x_2}{{P_2}@{i}}{C_i}}
\end{array}
$$
Notice here the operator $@$ is defined in $\ipc$. In the resulting notion of normal form, the conclusion of an elimination of disjunction can only be a disjunction, absurdity, or a variable. This is a restriction known not to break the completeness of the calculus \cite{FerrariFiorentiniJAR2019}.

\section{Final remarks}\label{sec:final}


The embedding of $\ipc$ into second-order logic is stable at the level of formulas, but has many variants at the level of proofs (and proof terms). While formulas are always translated using the second-order definitions of disjunction and absurdity, proofs can be translated either as implicitly done in \cite{Prawitz65} (see \cite{EspiritoSantoFerreira2021}), making full use of the elimination rule of the second-order quantifier; or as successive translation into $\fat$: as in \cite{FerreiraFerreira2009,FerreiraFerreira2013}, making use of instantiation overflow; as in \cite{PistoneTranchiniPetrolo2021}, optimizing the previous idea; as in our \cite{EspiritoSantoFerreira2020}, making use of the admissibility in $\fat$ of the elimination rules for disjunction and absurdity; or, finally, as in the present paper, optimizing the previous idea. In this spectrum of translations into $\fat$, increasingly better simulations of the commuting conversions are achieved, ending in their complete elimination obtained here.

Besides the elimination of commuting conversions, how good is the representation of $\ipc$ into $\fat$ induced by the new translation proposed here? The first question to answer in this direction is that of faithfulness of the translation, but that may require a separate paper as \cite{FerreiraFerreiraSL2015}. The question of preservation of normal forms, briefly touched in Section \ref{sec:discussion}, deserves a second look, maybe not necessarily in connection with the translation introduced here.

A final word about methodology. As in our previous papers \cite{EspiritoSantoFerreira2020,EspiritoSantoFerreira2021}, our development relies on the use of proof terms. It seems to us such choice had two decisive advantages. The first is that the highly bureaucratic argument needed in Section \ref{sec:special} to deal with the commuting conversion $\pii$ would be practically impossible with a different choice of notation. The second is that the study of the embeddings of $\ipc$ into $\fat$ becomes a study of translations between two different $\lb$-calculi, and this suggests the import of techniques from computer science. Here we imported the technique of ``compile-time optimization''  from programming language theory, specifically the reduction ``on the fly'' of the administrative redexes \cite{PlotkinTCS1975}.



\bibliographystyle{elsarticle-num} 
\bibliography{bibrefs}





\end{document}
\endinput